\documentclass{statsoc}
\usepackage[a4paper, total={6.5in, 9.7in}]{geometry}
\usepackage[utf8]{inputenc}

\usepackage{amsmath}
\usepackage{hyperref}
\usepackage{graphicx}   
\usepackage{amssymb}
\newtheorem{theorem}{Theorem}

\usepackage{subfig}
\newtheorem{corollary}{Corollary}

\newcommand{\hbeta}{\boldsymbol{\hat{\beta}}}
\newcommand{\hbetaz}{\hat{\beta_0}}

\newcommand{\alphaa}{\tau}
\newcommand{\halphaa}{\hat{\alphaa}}
\newcommand{\halpha}{\hat{\alphaa}}
\newcommand{\X}{\boldsymbol{\mathrm{X}}}
\newcommand{\x}{\boldsymbol{\mathrm{x}}}
\newcommand{\z}{\boldsymbol{\mathrm{z}}}

\newcommand{\Z}{\boldsymbol{\mathrm{Z}}}
\newcommand{\A}{\boldsymbol{\mathrm{A}}}
\newcommand{\B}{\boldsymbol{\mathrm{B}}}
\newcommand{\DD}{\boldsymbol{\mathrm{D}}}
\newcommand{\y}{\boldsymbol{\mathrm{y}}}
\newcommand{\one}{\boldsymbol{1}}
\newcommand{\tr}{\boldsymbol{\mathrm{w}}}
\newcommand{\trn}{w_n}
\newcommand{\bbeta}{\boldsymbol{\beta}}
\newcommand{\ggamma}{\boldsymbol{\gamma}}
\newcommand{\hgamma}{\boldsymbol{\hat{\gamma}}}
\newcommand{\eepsilon}{\boldsymbol{\epsilon}}

\newcommand{\cov}{\mathrm{cov}}
\newcommand{\corr}{\mathrm{corr}}

\newcommand{\I}{\boldsymbol{\mathrm{I}}}
\newcommand{\E}{\mathrm{E}}
\newcommand{\zero}{\boldsymbol{0}}
\newcommand{\uu}{\boldsymbol{u}}
\newcommand{\UU}{\boldsymbol{\mathrm{U}}}
\newcommand{\vv}{\boldsymbol{v}}
\newcommand{\xtx}{\X^{i,t} \X^i}
\newcommand{\xtxc}{\Z^{i,t} \, \Z^i}
\newcommand{\xtxi}{(\X^{i,t} \X^i)^{-1}}
\newcommand{\pp}{\boldsymbol{p}}
\newcommand{\qq}{\boldsymbol{q}}
\newcommand{\RR}{\boldsymbol{\mathrm{R}}}

\newcommand{\rrho}{\boldsymbol{\rho}}
\newcommand{\PPhi}{\boldsymbol{\Phi}}
\newcommand{\g}{\boldsymbol{\mathrm{g}}}
\newcommand{\mmu}{\boldsymbol{\mu}}
\newcommand{\SSigma}{\boldsymbol{\Sigma}}
\newcommand{\PPsi}{\boldsymbol{\Psi}}
\newcommand{\bias}{\delta}
\newcommand{\variance}{\sigma^2}
\newcommand{\dd}{\boldsymbol{\mathrm{d}}}
\newcommand{\wcm}{\boldsymbol{\Lambda}}
\newcommand{\EE}[1]{\mathrm{E}[\, #1 \,]}

\title[Matching and linear regression]{
Combining matching and linear regression: Introducing a mathematical framework and software for simulations, diagnostics and calibration
}

\author[Mahani {\it et al.}]{Alireza S. Mahani}
\address{Scientific Computing Group, Sentrana Inc,
	Washington, DC,
	USA.}
\email{alireza.mahani@sentrana.com}
\author[Mahani {\it et al.}]{Mansour T. A. Sharabiani}
\address{National Heart and Lung Institute, Imperial College London, UK.}

\begin{document}

\begin{abstract}
Combining matching and regression for causal inference provides double-robustness in removing treatment effect estimation bias due to confounding variables. In most real-world applications, however, treatment and control populations are not large enough for matching to achieve perfect or near-perfect balance on all confounding variables and their nonlinear/interaction functions, leading to trade-offs. [this fact is independent of regression, so a bit disjointed from first sentence.] Furthermore, variance is as important of a contributor as bias towards total error in small samples, and must therefore be factored into the methodological decisions. In this paper, we develop a mathematical framework for quantifying the combined impact of matching and linear regression on bias and variance of treatment effect estimation. The framework includes expressions for bias and variance in a misspecified linear regression, theorems regarding impact of matching on bias and variance, and a constrained bias estimation approach for quantifying misspecification bias and combining it with variance to arrive at total error. Methodological decisions can thus be based on minimization of this total error, given the practitioner's assumption/belief about an intuitive parameter, which we call `omitted R-squared'. The proposed methodology excludes the outcome variable from analysis, thereby avoiding overfit creep and making it suitable for observational study designs. All core functions for bias and variance calculation, as well as diagnostic tools for bias-variance trade-off analysis, matching calibration, and power analysis are made available to researchers and practitioners through an open-source R library, \textit{MatchLinReg}.
\end{abstract}



\section{Introduction}\label{section-introduction}

Despite a rich body of literature on causal inference in observation studies, some important questions remain unanswered. Motivated by the desire to produce unbiased estimation of treatment effect (TE) in the presence of confounding variables, theoretical work in this field has largely focused on large-sample, bias-removal properties of techniques such as matching and regression~\citep{abadie2006large,abadie2011bias}. In applied settings, however, sample sizes are finite and often small (including control groups), asymptotic properties are not realized, and variance is as important as bias in contributing towards total estimation error. This is particularly true in medical applications such as studying the effectiveness of a novel in-patient heart procedure~\citep{kereiakes2000abciximab}, where not only randomized treatment assignment is unfeasible due to ethical or logistical concerns, but data collection is costly and time-consuming, and thus early identification of treatment effectiveness is highly valuable.

In particular, the general recognition by researchers of the added benefit of combining matching and regression to achieve the so-called `double-robustness'~\citep{rubin1973use,carpenter1977matching,rubin1979using,robins1995semiparametric,heckman1997matching,rubin2000combining,glazerman2003nonexperimental,abadie2006large,ho2007matching,abadie2011bias} has yet to translate into a set of guidelines and tools that allow for their effective utilization by practitioners. For example, what caliper size should be used in matching? What is the impact of the choice of matching technique, e.g. Mahalanobis distance matching vs. propensity score matching (PSM), on TE estimation? What terms should be included in matching (e.g. in a logistic regression model used for generating propensity scores)? How do such decisions affect not only TE estimation bias but also variance? At what point is study power diminished beyond usefulness when matching is applied more and more restrictively to reduce bias? Simulation studies have shed light on some of these questions empirically~\citep{austin2007comparison,austin2009some,gayat2012propensity}, and providing a theoretical foundation can facilitate generalizable conclusions from such empirical work~\citep{imbens2004nonparametric}.

In this paper, we develop the mathematical framework for quantifying the combined effect of matching and linear regression on TE estimation bias and variance in finite samples with arbitrary covariate distributions. Similar to the approach of \cite{ho2007matching}, we consider matching as a non-parametric pre-processing step prior to regression adjustment. We cast the model misspecification problem as a covariate omission problem, and derive closed-form expressions for TE estimation bias and variance in linear regression adjustment. Our expressions for bias and variance are useful in three capacities:
\begin{enumerate}
\item Theory: Using the normalized form of bias expression, we prove that perfect matching on included and omitted covariates eliminates bias. Furthermore, perfect matching on included covariates minimizes variance, given a fixed number of treatment and control observations. The variance-minimization property of matching provides theoretical support for the notion that study power is not lost as rapidly when data loss is due to matching, compared to data loss that is random.
\item Simulations: The closed-form expressions for bias and variance can be efficient and accurate substitutes for Monte Carlo simulations for studying matching and linear regression. We have used these expressions for most of the simulations presented in this paper, and expect that future empirical studies in the field will benefit from this efficient simulation tool.
\item Diagnostic and calibration software: Our expressions can be utilized in a diagnostic capacity to analyze bias-variance trade-off resulting from various choices of matching and regression parameters, and to quantify the impact of matching on study power. By utilizing such diagnostic tools - embedded in the open-source R package MatchLinReg~\citep{mahani2015matchlinreg} - practitioners will be able to make better-informed decisions and the quality of causal inference in observational studies will be enhanced as a result.
\end{enumerate}

The rest of this paper is organized as follows. We continue by reviewing previous work on combining matching and regression for causal inference, introduce data sets and software used in the paper, and present simulation results using real data sets to highlight the benefits and complexities associated with combining matching and regression, which motivated our research. In Section~\ref{section-framework}, we present our mathematical framework. We outline the scope and philosophy of the paper, formally define the problem, derive equations for TE bias and variance estimation using linear regression, and offer theorems regarding the impact of matching on bias and variance. In Section~\ref{section-tools}, we use the developed framework to create diagnostic and calibration tools. We begin with exploratory diagnostics such as relative squared bias reduction due to single omitted covariates, and comparison of normalized single-covariate biases, and proceed to generate an aggregate measure of normalized bias using constrained bias estimation methodology. We then combine normalized bias and variance to arrive at normalized MSE, using a conversion factor called `omitted R-squared'. Minimizing MSE allows us to select the matching parameters, a process which we refer to as calibration. We close this section by reviewing the open-source R package, MatchLinReg, that contains the core framework as well as the diagnostic and calibration tools. Finally, in Section~\ref{section-discussion} we provide a summary of contributions in this paper, and outline pointers for promising future research directions.

A list of mathematical symbols used throughout the paper can be found in Appendix~\ref{appendix-math-symbols}. Frequently-used acronyms are listed in Appendix~\ref{appendix-acronyms}. Details of mathematical derivations and theorem proofs are provided in the remaining appendices.

\subsection{Previous work}\label{subsection-litrev}
There is a rich body of literature on causal inference for observational studies, and some have focused - partially or wholly - on combined use of matching and regression adjustment. \cite{rubin1973use} derive expressions for bias reduction for a single confounder, considering parallel and non-parallel response surfaces. \cite{carpenter1977matching} analyze the performance of nearest-neighbor matching, and its combination with regression adjustment. While their primary focus is on normally-distributed covariates, they present an Analysis-of-covariance (ANCOVA) equation for variance, which is equivalent to our standard expression (Eq.~\ref{eq-variance-standard}). \cite{rubin1979using} study the effect of various approaches for combining matching and regression adjustment for controlling bias, and conclude that covariate matching combined with regression on paired differences performs best. \cite{rubin1992affinely,rubin1996matching} study bias reduction resulting from propensity score matching for normally-distributed covariates. Their expressions rely on correlation between covariates and outcome. [not all of them; edit this sentence] \cite{rubin2000combining} study the bias-removal impact of combining propensity score matching with regression adjustment and conclude that the combination has superior performance than regression alone. \cite{ho2007matching} use real data sets to illustrate that using matching as a non-parametric pre-processing step reduces model dependence in parametric  causal inference. While cautioning that increased variance from data loss can outweigh reduced bias in matching, they also point out that `no precise rule exists for how to make these [bias-variance tradeoff] choices.' \cite{schafer2008average} use simulated data based on Add Health study~\citep{udry1998national} to compare the performance of various combinations of ANCOVA and propensity score matching techniques. \cite{hosman2010sensitivity} present methods for sensitivity analysis of regression, potentially combined with propensity-score-based stratification, to confounder omission. \cite{abadie2011bias} prove that adding a bias-correcting step such as least-squares regression to nearest-neighbor matching renders it $N^{1/2}$-consistent, but their analysis is asymptotic in nature, and is focused on matching with replacement which is typical of econometric literature.


\subsection{Setup}\label{subsection-setup}

We use the statistical software R~\citep{R2015language} for all simulations presented in this paper. Two observational data sets form the basis of our simulations: lalonde~\citep{lalonde1986evaluating} and lindner~\citep{kereiakes2000abciximab}. Both data sets are available in the R package twang~\citep{ridgeway2015twang}. The lalonde data set was collected to study the effectiveness of job training programs on future earnings of 614 participants (185 treatment, 429 control), with 8 adjustment covariates (4 numeric, 4 binary). The lindner data set describes an observational study of the impact of Percutaneous Coronary Intervention (PCI) on cardiac-related costs for 996 patients (698 treatment, 298 control), given 7 covariates (3 numeric, 4 binary).

Throughout the examples, we use PSM for matching (without replacement), using the R package Matching~\citep{matching2011package}. Furthermore, as a concrete illustration of a calibration problem, we focus on identifying the caliper size used in matching (on propensity scores). However, our framework can be applied equally well to other matching techniques - such as Mahalanobis distance matching - as well as calibration problems, such as determining what terms to include in propensity score model, whether to use matching with or without replacement, etc.

All simulation results presented in Section~\ref{section-tools} utilize the closed-form expressions developed in Section~\ref{section-framework}. All the core functions, as well as the diagnostic and calibration methods described in this paper are all available as part of an open-source R package, MatchLinReg~\citep{mahani2015matchlinreg}. See Section~\ref{subsection-software} for an overview of MatchLinReg.

\subsection{Complexities of TE estimation: Looking beyond double-robustness}\label{subsection-beyond-double}
A common approach for combining matching and regression in statistical analysis is to use matching as a non-parametric pre-processing step before regression~\citep{ho2007matching}. This approach allows practitioners to apply the familiar regression tools to their problem once matching is done, although it has been argued that issues such as calculation of standard errors must take into consideration the matching step~\citep{austin2007propensity}. In linear regression with a treatment indicator binary variable and other (adjustment) covariates included, the coefficient of treatment indicator equals the treatment effect. Including adjustment covariates in regression helps remove residual bias due to imperfect matching. A special case is when no adjustment covariates are used in regression; we call this method the `simple difference' method, since it is equivalent to calculating the mean outcome in treatment and control groups and subtracting them to obtain TE. As an alternative view, matching helps reduce sensitivity of regression adjustment to model misspecification (Section~\ref{subsection-matching-impact}). This mutually-reinforcing effect of combining matching and regression for TE estimation has been referred to as `double-robustness' in the literature~\citep{stuart2010matching}.

Figure~\ref{fig-double-robustness} illustrates the combined effect of matching and regression for TE estimation, using Monte Carlo simulations based on lalonde and lindner data sets. For each data set, noise variance as well as coefficients of original variables were extracted from linear regression on full data sets, and a quadratic term ($re74^2$ for lalonde and $ejecfrac^2$ for lindner) was added to the generative model, with its coefficient chosen so as to achieve an omitted R-squared ($R_o^2$) of 5.9\% and 3.3\% for lalonde and lindner, respectively (40\% of $R^2$ from regression on main effects.). $R_o^2$ quantifies the percent of variation in outcome attributable to covariates that are not included in the regression model. See Section~\ref{subsection-quantify-bias} for mathematical definition. Regressions used in Monte Carlo simulations did not include this quadratic term, in order to simulate regression misspecification through covariate omission (Section~\ref{subsection-problem-definition}). As expected, matching and regression each significantly reduce bias (compared to simple difference method) in both data sets, and combining them further reduces bias.

\begin{figure}
\centering
\begin{tabular}{cc}
\subfloat{\includegraphics[height=3.in]{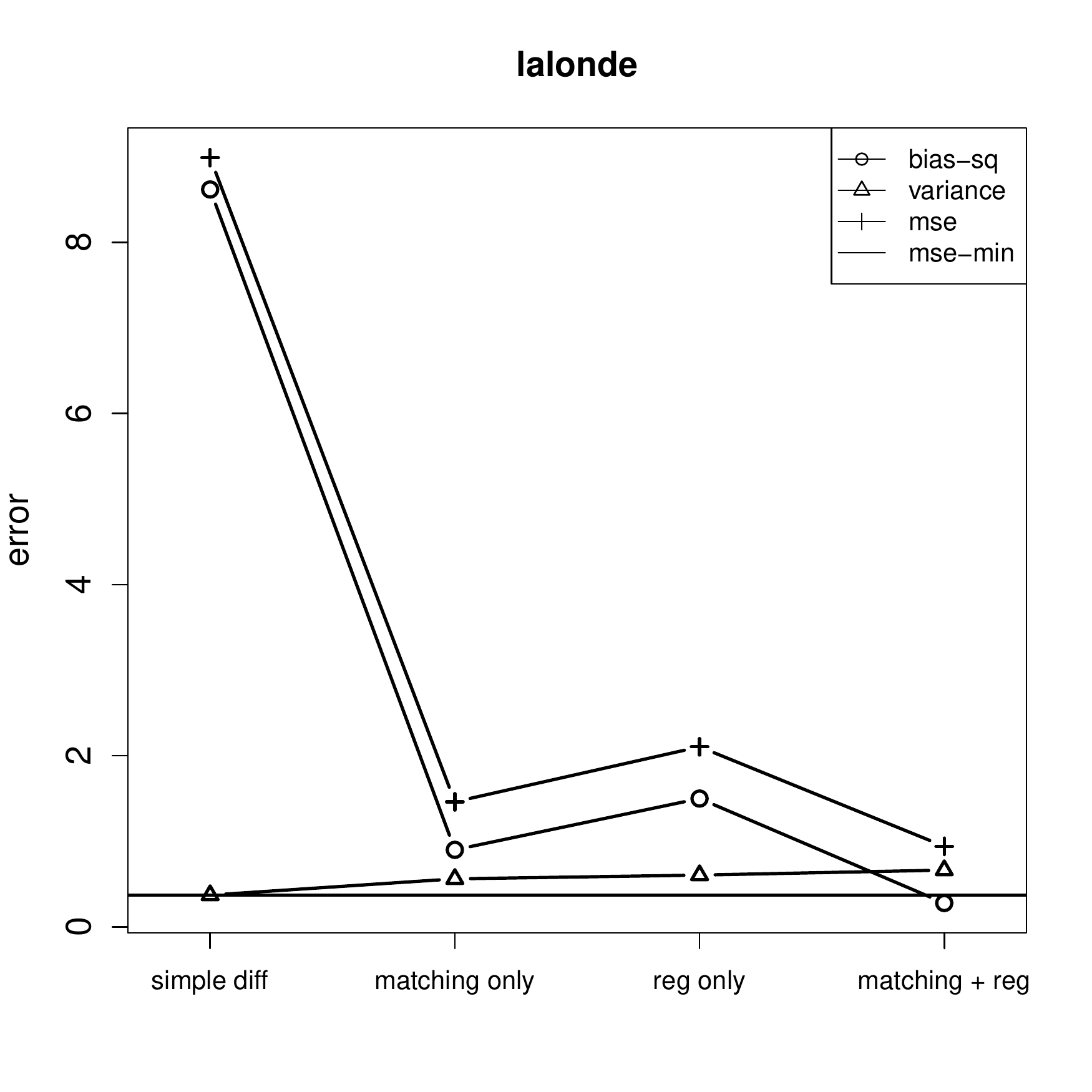}} & 
\subfloat{\includegraphics[height=3.in]{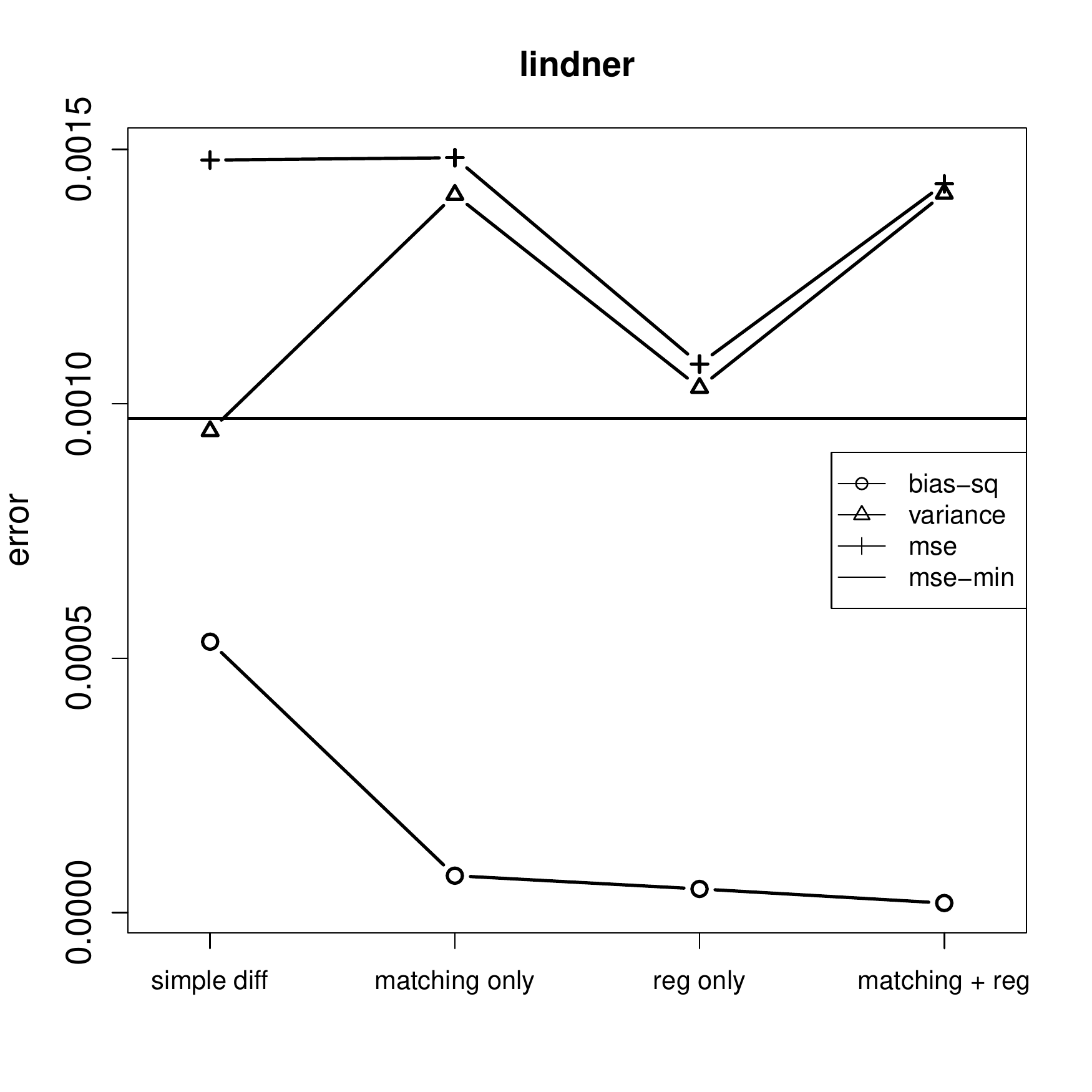}} \\ 
\end{tabular}
\caption{Combining matching and regression for TE estimation in lalonde (left) and lindner (right) data sets. For each data set, error components (squared bias, variance, MSE) are compared using four approaches: 1) simple difference, 2) matching followed by simple difference, 2) regression adjustment, and 3) matching followed by regression adjustment. Horizontal line indicates minimum achievable MSE. Matching was done on propensity scores (on linear scale) resulting from logistic regression of treatment indicator on included covariates and (omitted) quadratic term, using a caliper of 1.5 (lalonde) and 0.2 (lindner). R package `Matching' was used for matching. Numbers are based on 10,000 Monte Carlo simulation. For lalonde, outcome is \textit{re78} divided by 1000, while for lindner, outcome is the logarithm of \textit{cardbill}. Omitted covariate is $re74^2$ for lalonde and $ejecfrac^2$ for lindner, with coefficient of this omitted covariate chosen so as to achieve an $R_o^2$ of 5.9\% and 3.3\% for lalonde and lindner, respectively.}
\label{fig-double-robustness}
\end{figure}

However, mean squared error (MSE) is the sum of squared bias and variance:
\begin{equation}
MSE = (bias)^2 + variance
\end{equation}
While matching reduces TE bias in regression adjustment caused by covariate omission, yet it does so by discarding data points and thus reducing data size, leading to increased variance. As Figure~\ref{fig-double-robustness} shows, while in lalonde data set the bias-reduction effect of adding matching to regression more than offsets the increased variance such that MSE is reduced, the net effect for lindner is adverse, such that regression alone is the best option. This represent one of the key challenges in calibrating the matching-plus-regression combination, i.e. striking the right balance between bias reduction and variance increase so as to minimize MSE for TE. Also, it must be noted that despite the adverse effect of matching on variance, matching discards data quite efficiently and therefore increases variance slower than randomly discarding data (Theorem~\ref{theorem-variance-no-replacement}).

Variance considerations are not the only source of complexity. Matching affects bias indirectly and only by adjusting the distribution of covariates across treatment and control groups. In determining the ultimate impact of matching on TE bias, we face several complexities: 1) Bias induced by an omitted covariate is not a monotonic function of its imbalance (Figure~\ref{fig-bias-complexity}, top row); 2) matching does not improve balance for all covariates, although it tends to do so for the most imbalanced covariates (Figure~\ref{fig-bias-complexity}, middle row); 3) as a result of 1 and 2, matching does not improve bias due to all omitted covariates, but again it tends to reduce the largest biases (Figure~\ref{fig-bias-complexity}, bottom row); 4) Finally, in generating Figure~\ref{fig-bias-complexity} we assumed only one omitted covariate at a time, but in reality multiple such covariates can be missing from our regression model. What combination of these unknown terms, and their impact on TE bias, should be considered in how we combine regression and matching? The mathematical framework and diagnostics software developed in this paper seek to address the above challenges.

Note that Equal-Percent-Bias-Reduction (EPBR) property of matching techniques~\citep{rubin1976multivariate} such as PSM does not solve the complexities associated with TE bias: 1) EPBR is valid only when covariate distributions follow restrictions such as being multivariate normal, 2) EPBR only works for linear combinations of covariates, and not their nonlinear functions such as interactions and powers, which are our main candidates for omitted covariates, thanks to `ignorability of treatment assignment' assumption (Section~\ref{subsection-philosophy}).

\begin{figure}
\centering
\begin{tabular}{cc}
\subfloat{\includegraphics[height=2.5in]{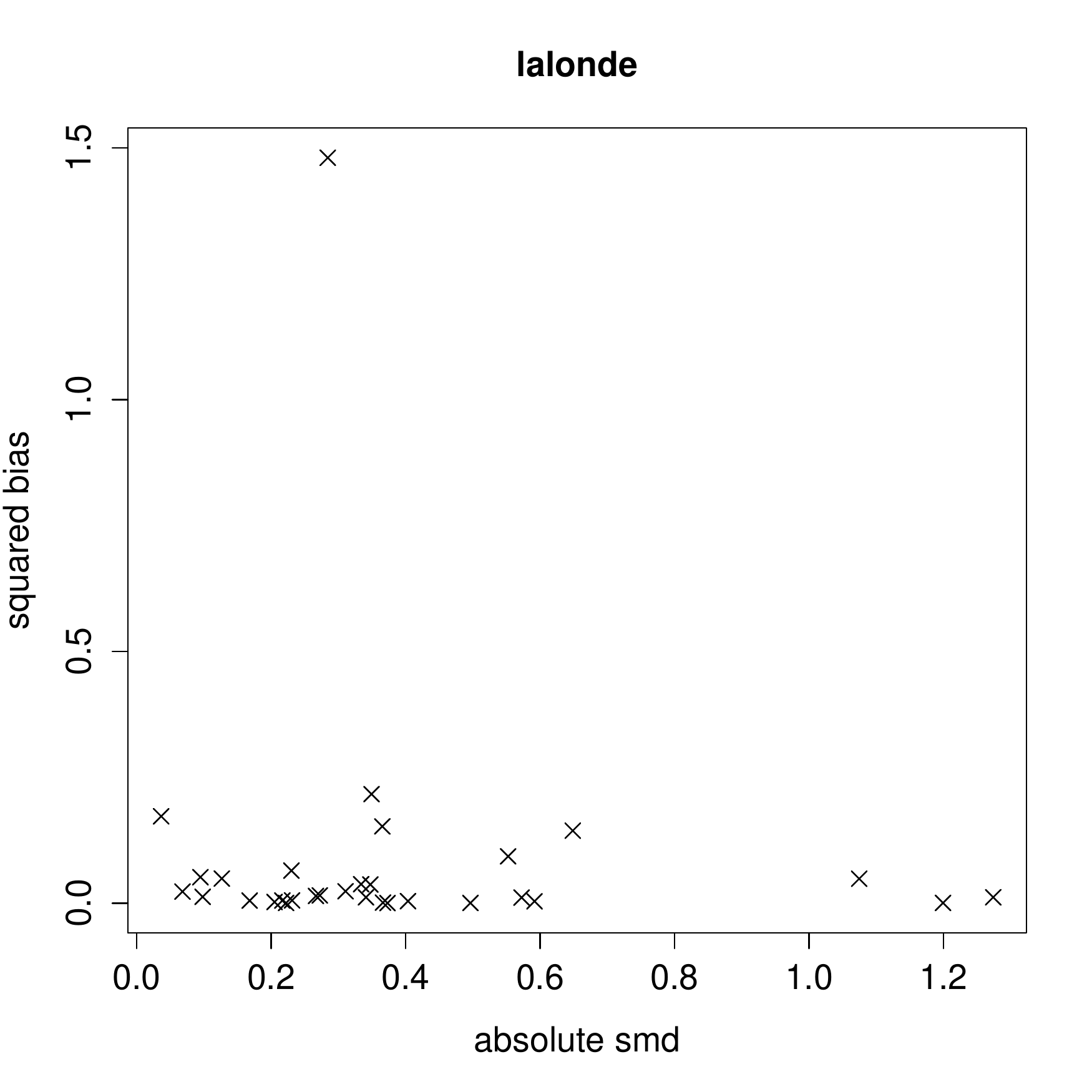}} & 
\subfloat{\includegraphics[height=2.5in]{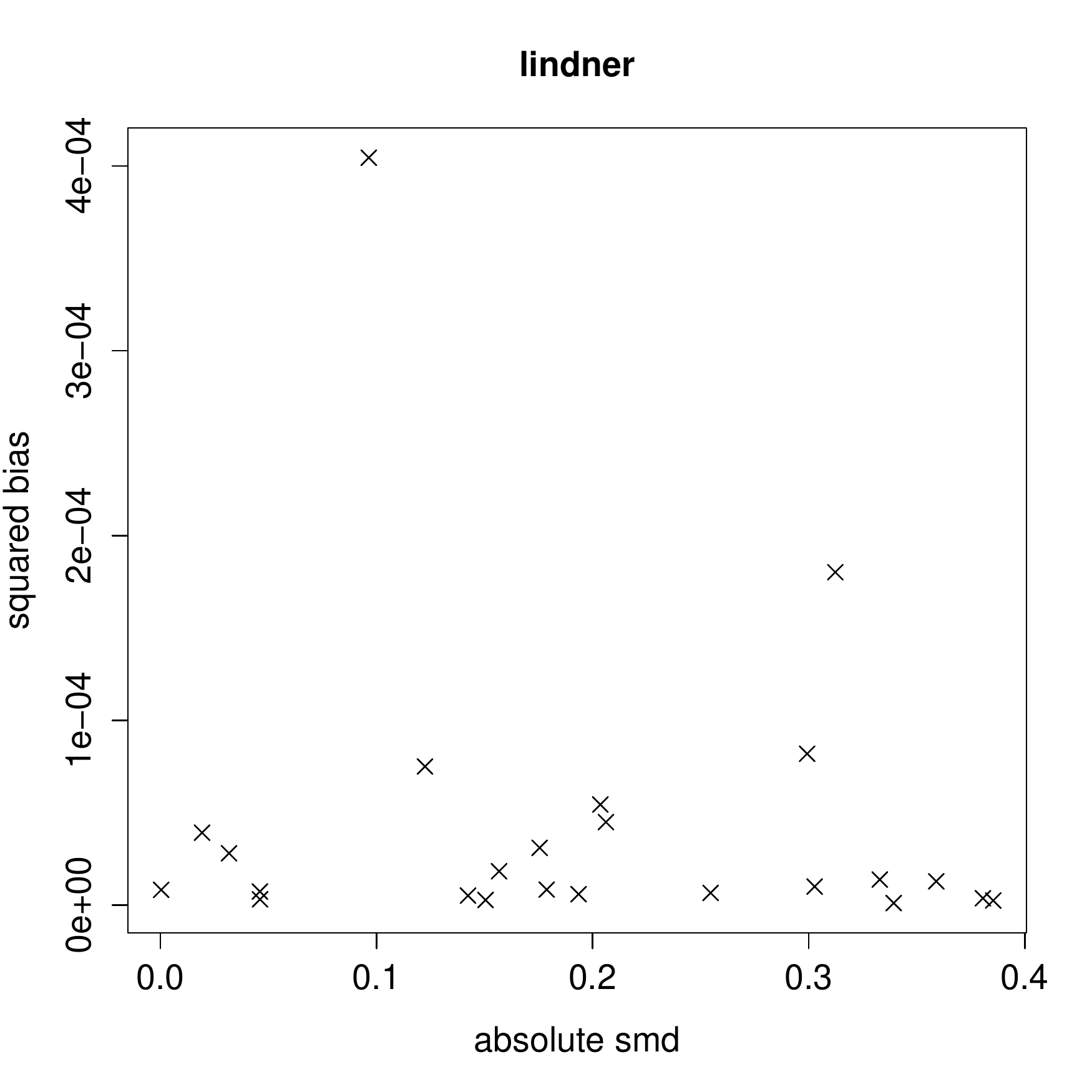}} \\ 
\subfloat{\includegraphics[height=2.5in]{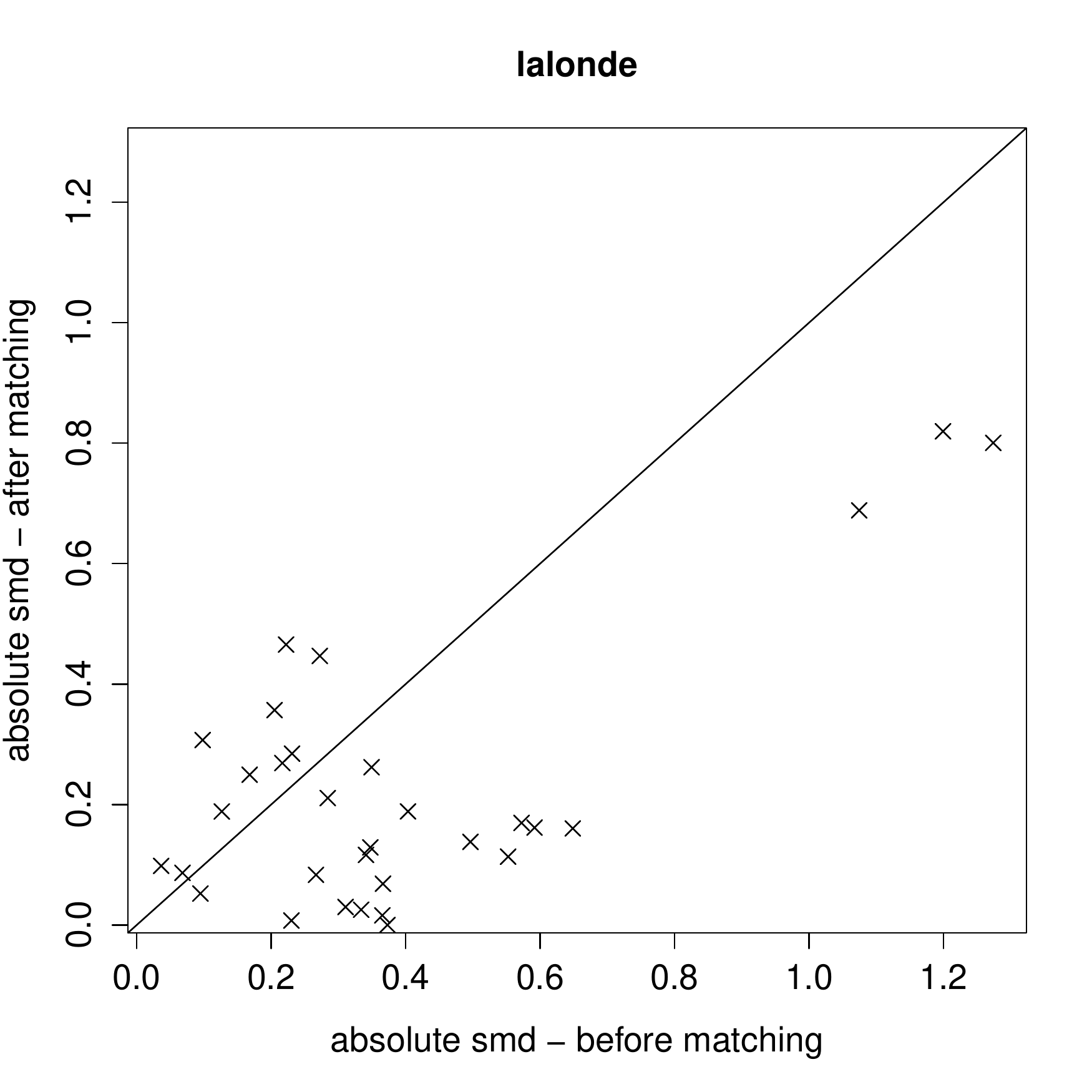}} & 
\subfloat{\includegraphics[height=2.5in]{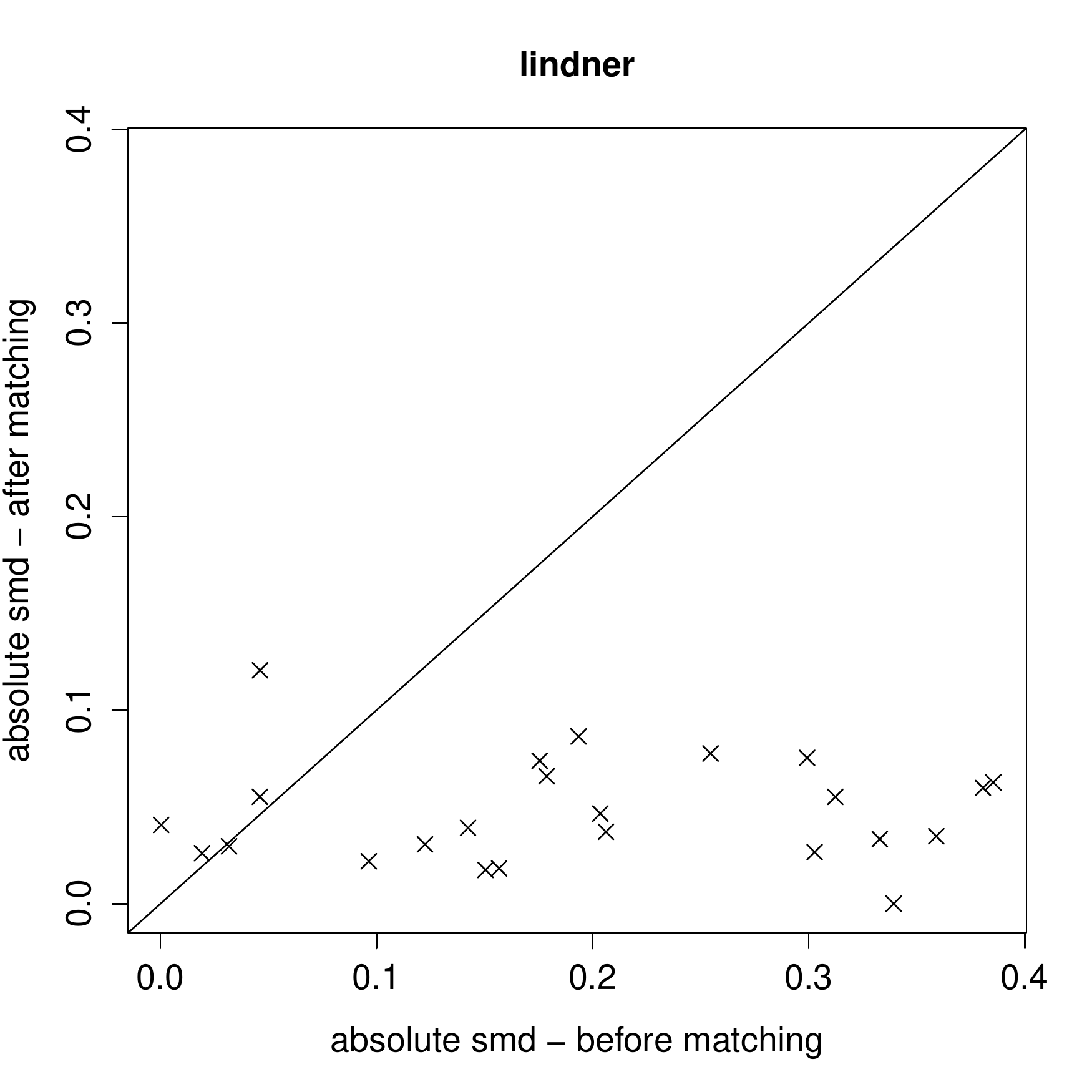}} \\
\subfloat{\includegraphics[height=2.5in]{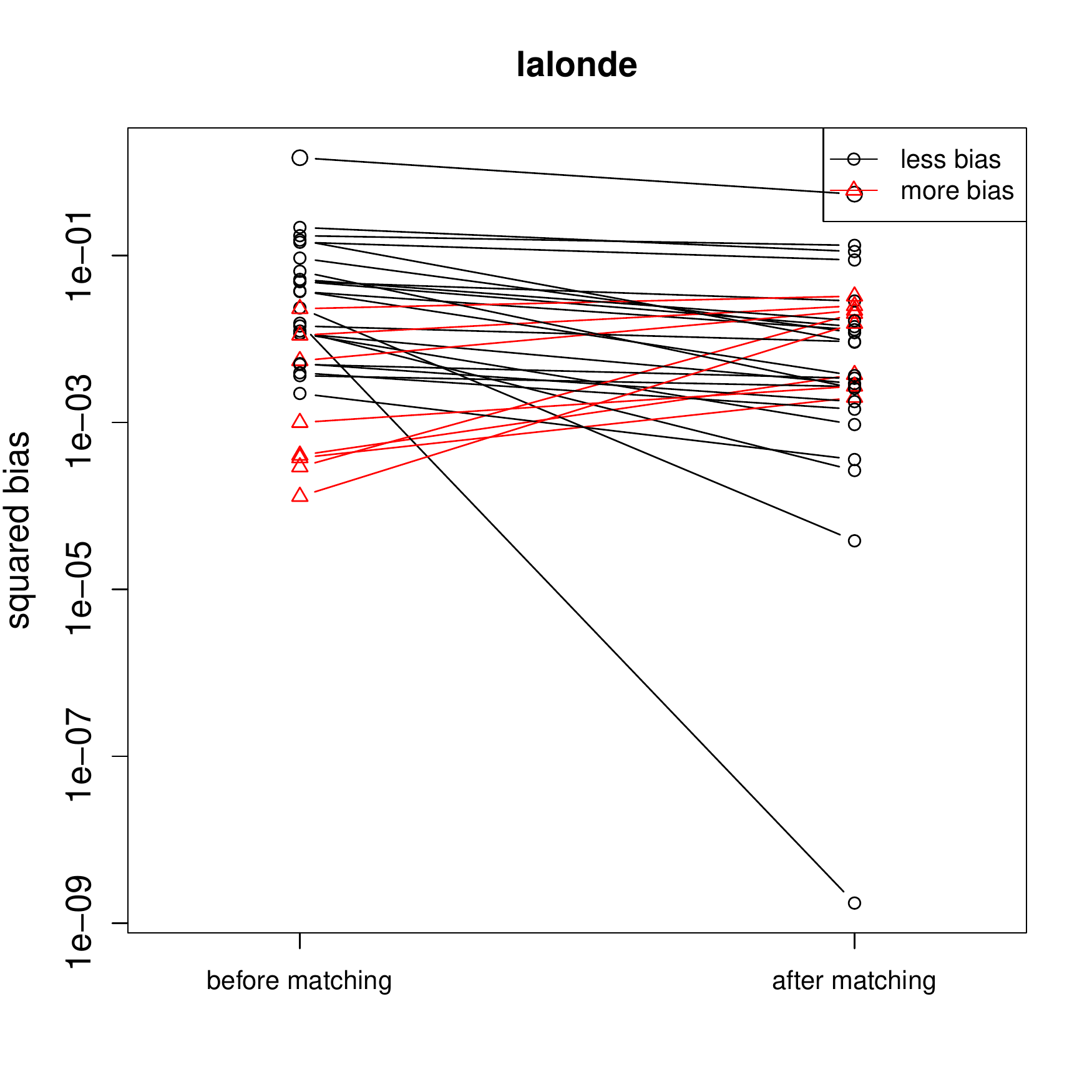}} &
\subfloat{\includegraphics[height=2.5in]{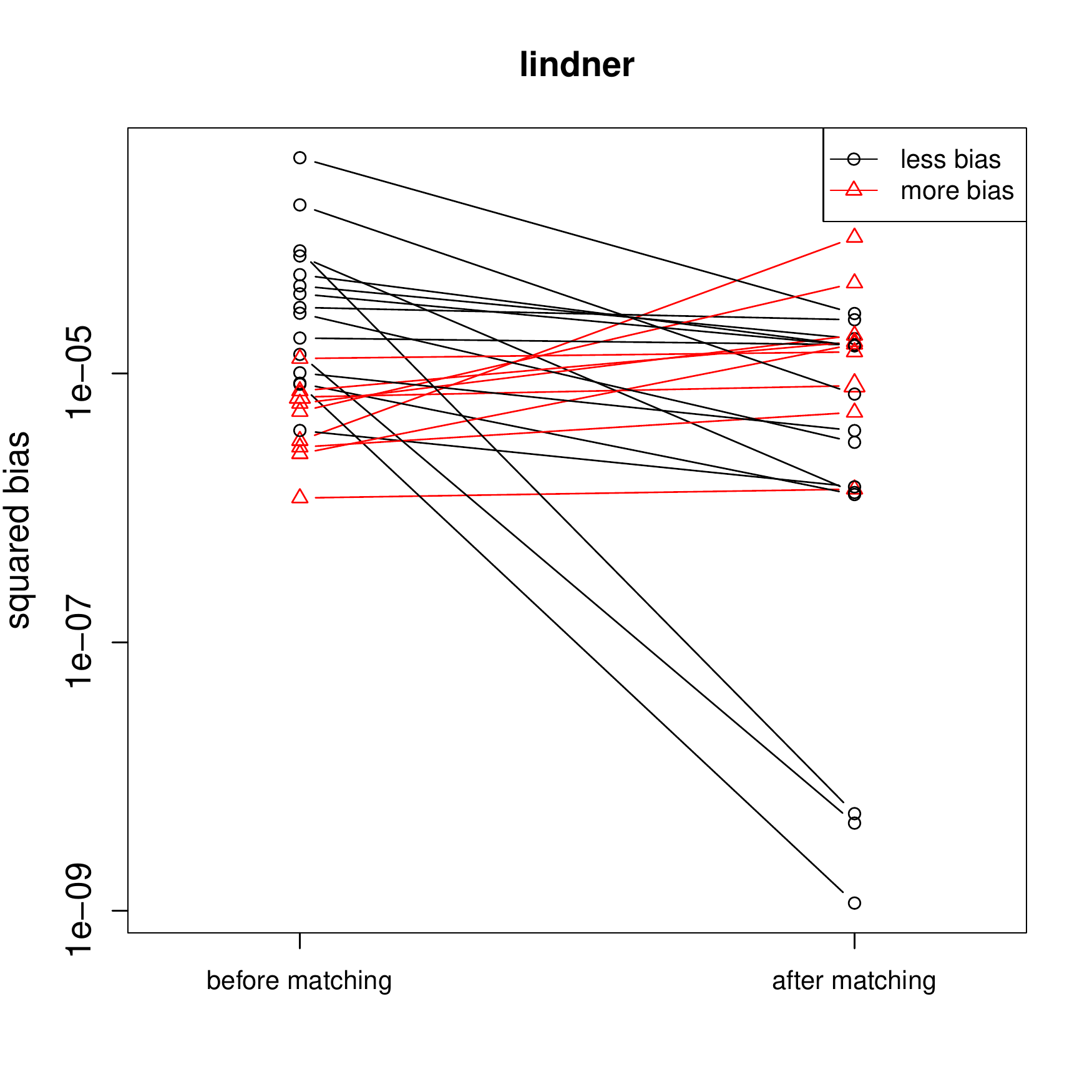}} 
\end{tabular}
\caption{Complex relationship between matching, covariate imbalance, and omission bias for lalonde (left) and lindner (right) data sets. Selection of coefficients for omitted covariates, and regression and matching methods and parameters are identical to Figure~\ref{fig-double-robustness}. Equations of Section~\ref{subsection-bias-var-equations} are used for bias calculations. Top row: squared bias vs. absolute SMD for all second-order interaction terms (missing from regression model). Middle row: Comparison of SMD before and after matching for omitted terms. Bottom row: Comparison of square bias, before and after matching, for omitted terms.}
\label{fig-bias-complexity}
\end{figure}

\section{Framework}\label{section-framework}
In this section, we develop the mathematical framework for analyzing the combined effect of matching and linear regression for TE estimation. We begin by defining the assumptions and notation used in the paper.

\subsection{Approach and assumptions}\label{subsection-philosophy}

\textit{Ignorability of treatment assignment} Also referred to as `unconfoundedness'~\citep{imbens2004nonparametric} and `conditional independence'~\citep{lechner1999earnings,lechner2002program}, this assumption states that, conditioned on the covariates available in our data set, treatment assignment is random~\citep{rosenbaum1983central}. In the context of a standard linear model (see below), this assumption means that all covariates in the generative model are derivatives (e.g. powers, interaction terms, splines, etc) of a core set of covariates included in the data set. This, in turn, means that perfect matching on the core set of covariates would automatically lead to perfect matching on the full set of covariates in the generative model.

\textit{Matching as pre-processing for regression adjustment}
Similar to \cite{ho2007matching} [more examples], we use matching as a pre-processing step before a parametric estimation of TE using linear regression. Combined with the parallel response surface assumption (see below), linear regression on the full data set produces TE as the coefficient of treatment indicator variable~\citep{imbens2004nonparametric}. Our approach can be considered an extension of two-group ANCOVA~\citep{rubin1973use,carpenter1977matching,quade1982nonparametric} to include the impact of regression misspecification parameterically and through the covariate omission concept. Notable alternatives include regression on matched pair differences~\citep{rubin1979using}, and separate regressions on treatment and control groups~\citep{imbens2004nonparametric}. Using linear regression for parametric TE estimation allows us to derive closed-form expressions for TE bias and variance (Section~\ref{subsection-bias-var-equations}) for finite samples and without any parametric assumptions about the distribution of covariates in the treatment and control groups, or how their corresponding covariate matrices are related. Matching followed by a simple-difference calculation of TE, i.e. without regression adjustment, is a special case of the above framework, but with no adjustment covariates included in the regression model. In this case, the regression model has only two coefficients: intercept and TE.

\textit{Parallel response surfaces}
This assumption states that the difference between functions describing mean outcome for treatment and control groups is independent of the covariate vector. In other words, TE is constant for all observations, conditional and marginal/average TE are the same, and independent of sample. In a linear regression setting, this assumption means the generative model does not contain any interaction terms between treatment indicator and adjustment covariates, and correspondingly we do not include any such interaction terms in the regression model. While an important restriction, the parallel response surface assumption allows us to focus on developing the mathematical framework without concern for what relative weight to attach to error in calculating TE for different strata.

\textit{Covariate omission as model misspecification}
Matching reduces potential TE estimation bias caused by a misspecified regression model. In the linear regression framework - where terms can be nonlinear in original covariates - misspecification manifests itself in the form of covariate omission, i.e. one or more terms included in data generation being absent from the regression model. Importantly, including nuisance terms in regression, i.e. that were not part of the generative model, does not induce bias, but only increases variance.

\textit{Exclusion of outcome variable from analyses}
It is possible to use the outcome variable - implicitly or explicitly - during the matching calibration process. For example, one may focus on achieving better covariate balance for those covariates that show a stronger correlation with outcome. While careful use of outcome variable can enhance the quality of entire modeling process including matching, in the hands of inexperienced practitioners it can be lead to overfit models with poor generalizability in real-world~\citep{babyak2004you}. We have therefore chosen to develop our diagnostic and calibration tools (Section~\ref{section-tools}) with complete exclusion of outcome variable from consideration. In addition to preventing overfit creep, this approach also makes our framework and tools ideal for cases where outcome data is not available yet and matching is used to select subjects for follow up~\citep{stuart2010matching}.


\subsection{Problem definition}\label{subsection-problem-definition}

In this paper, we focus on causal inference for a \textit{Standard Linear Model} (SLM), which consists of two components: a data generation (generative) model and a coefficient estimation (regression) model:
\begin{enumerate}
\item Generative model: Outcome ($\y$) is continuous and its mean is a linear function of covariates ($\X$):
\begin{align}
\EE \y &= \X \bbeta. \label{eq-genmodel-mean}
\end{align}
Data generation has an i.i.d. noise distribution:
\begin{align}
\cov[\y] = \sigma_0^2 \, \I. \label{eq-genmodel-var}
\end{align}
Data-generation covariates consist of a binary treatment indicator ($\tr$), a unit vector ($\one$), as well as a set of adjustment covariates ($\Z$), none of which include interactions with treatment indicator (parallel response surface):
\begin{align}
\X &= \left[ \begin{array}{ccc} \tr & \one & \Z \end{array} \right]. \label{eq-def-X}
\end{align}
$\tr$ is 1 for treatment observations and 0 for controls. Correspondingly, the vector of coefficients ($\bbeta$) consists of TE ($\tau$), intercept ($\beta_0$) and vector of coefficients for adjustment covariates ($\ggamma$):
\begin{align}
\bbeta \equiv \left[ \begin{array}{ccc} \tau & \beta_0 & \ggamma^t \end{array} \right]^t. \label{eq-def-beta}
\end{align}
For SLM, average TE for treated (ATT), sample average TE (ATE) and average TE for controls (ATC) are all the same and equal to the coefficient of treatment indicator variable, i.e. $\tau$. We refer to this entity simply as TE (treatment effect).
\item Regression model: Ordinary-Least-Squares (OLS) is used to estimate model coefficients, including $\tau$. However, only a subset of covariates ($\X^i$) are included in the regression model. This subset includes $\tr$ and $\one$, as well as $K^i$ adjustment covariates ($\Z^i$), which we call `included adjustment covariates' or simply included covariates. The remaining $K^o$ adjustment covariates ($\Z^o$) are called `omitted adjustment covariates', or simply omitted covariates:
\begin{align}
\Z &\equiv \left[\ \begin{array}{cc} \Z^i & \Z^o \end{array} \right], \label{eq-def-Z} \\
\X^i &\equiv \left[ \begin{array}{ccc} \tr & \one & \Z^i \end{array} \right], \label{eq-def-Xi} \\
\X &= \left[\ \begin{array}{cc} \X^i & \Z^o \end{array} \right]. \label{eq-X-Xi-Zo}
\end{align}
The vector of coefficients for adjustment covariates is correspondingly partitioned into included ($\ggamma^i$) and omitted ($\ggamma^o$) vectors:
\begin{align}
\ggamma &= \left[ \begin{array}{cc} \ggamma^{i,t} & \ggamma^{o,t} \end{array} \right]^t, \label{eq-decompose-gamma} \\
\bbeta^i &\equiv \left[ \begin{array}{ccc} \tau & \beta_0 & \ggamma^{i,t} \end{array} \right]^t, \label{eq-decompose-betai} \\
\bbeta &= \left[ \begin{array}{cc} \bbeta^{i,t} & \ggamma^{o,t} \end{array} \right]^t. \label{eq-beta-betai-gammao}
\end{align}
Ordinary-Least-Squares (OLS) estimation of coefficients leads to:
\begin{align} \label{eq-ols}
\hbeta &= \xtxi \, \X^{i,t} \y,
\end{align}
with
\begin{align} \label{eq-def-hbeta}
\hbeta \equiv \left[ \begin{array}{ccc} \hat{\tau} & \hat{\beta_0} & \hat{\ggamma^i}^t \end{array} \right]^t.
\end{align}
Note that omitted covariates ($\ggamma^o$) are not estimated, since they are not included in the regression model.
\end{enumerate}

\subsection{TE Bias and variance equations}\label{subsection-bias-var-equations}

As mentioned earlier, TE estimation bias in regression adjustment is caused by model misspecification, which is manifested as covariate omission in our framework, i.e. when $\Z^o \neq \zero$, or equivalently when $\X \neq \X^i$. It is easy to check that a correctly-specified model is unbiased:
\begin{equation}
\X^i = \X \,\,\, \Longrightarrow \,\,\, \EE \hbeta= (\X^t \X)^{-1} \, \X^t \, \EE \y = (\X^t \X)^{-1} \, \X^t (\X \bbeta) = \bbeta,
\end{equation}
where we have combined Eqs.~\ref{eq-genmodel-mean} and \ref{eq-ols}. With missing covariates, we have:
\begin{equation}\label{eq-bias-vector}
\EE {\hbeta} - \bbeta^i = (\X^{i,t} \X^i)^{-1} \X^{i,t} \left( \X^i \bbeta^i + \Z^o \ggamma^o \right) - \bbeta^i = (\X^{i,t} \X^i)^{-1} (\X^{i,t} \Z^o) \, \ggamma^o.
\end{equation}
As for variance, we have the following standard expression:
\begin{equation}\label{eq-covar-matrix}
\cov[\hbeta] = \sigma_0^2 \, (\X^{i,t} \X^i)^{-1}.
\end{equation}
(In the case of repeated observations, e.g. in matching with replacement, the bias expression remains valid, while the variance expression must be modified. See Appendix~\ref{appendix-var-repeat}.) Given our convention regarding the order of covariates in Eq.~\ref{eq-def-hbeta}, TE bias and variance are the first / top elements of the bias vector / covariance matrix, respectively, given in Eqs.~\ref{eq-bias-vector} and \ref{eq-covar-matrix}.

Next, we rewrite the above expressions in a way that highlights the role of covariate imbalance and the contribution from matching. Since Eqs.\ref{eq-bias-vector} and \ref{eq-covar-matrix} both contain $\xtxi$, we begin by transforming this matrix. From Equation~\ref{eq-def-Xi}, we have:
\begin{eqnarray}
\xtx &=& \left[ \begin{array}{c} \tr^t \\ \one^t \\ \Z^{i,t} \end{array} \right] \left[ \begin{array}{ccc} \tr & \one & \Z^i \end{array} \right] \\
&=& \left[ \begin{array}{ccc} N_t & N_t & \tr^t \Z^i \\ N_t & N_t + N_c & \one^t \Z^i \\ \Z^{i,t} \tr & \Z^{i,t} \one & \Z^{i,t} \Z^i \end{array} \right], \label{eq-xtx}
\end{eqnarray}
where we have taken advantage of $\tr^t \tr = \one^t \tr = \tr^t \one = N_t$ and $\one^t \one = N_t + N_c$. Since $\xtx$ is symmetric, so is its inverse. To calculate TE bias and variance, we are only concerned with the first row/column of $\xtxi$. In Appendix~\ref{appendix-xtxi}, we prove that:
\begin{align}\label{eq-xtxi}
\xtxi = \left( \begin{array}{ccc} \frac{1}{N_t} + \frac{1}{N_c} + \uu^{i,t} \A^{-1} \uu^i & -\frac{1}{N_c} + \frac{1}{N_c} \, (\pp - \qq)^t \A^{-1} \, \uu^i & \uu^{i,t} \A^{-1} \\ -\frac{1}{N_c} + \frac{1}{N_c} \, (\pp - \qq)^t \A^{-1} \, \uu^i & \hdots & \hdots \\  \A^{-1} \, \uu^i & \hdots & \hdots \end{array} \right),
\end{align}
where $\uu^i$ is the vector of mean differences of included covariates, $\pp$ and $\qq$ are vector of sums of included covariates across treatment/all observations, and $\A$ is the pooled, within-group covariance matrix for included covariates:
\begin{align}
\uu^i &\equiv \frac{1}{N_c} \, (\one - \tr)^t \, \Z^i - \frac{1}{N_t} \, \tr^t \, \Z^i, \label{eq-def-ui} \\
\pp &\equiv \Z^{i,t} \tr,\label{eq-def-p} \\
\qq &\equiv \Z^{i,t} \one,\label{eq-def-q} \\
\A &\equiv (N_t - 1) \, \cov(\Z^i)_T + (N_c - 1) \, \cov(\Z^i)_C. \label{eq-def-A}
\end{align}
Within-group covariance matrices are formally defined as follows:
\begin{align}
\cov(\Z^i)_T &\equiv \frac{1}{N_t - 1} \Z_T^{i,t} \, (\I_T - \frac{1}{N_t} \one_T \one_T^t) \, \Z_T^i, \label{eq-cov-T} \\
\cov(\Z^i)_C &\equiv \frac{1}{N_c - 1} \Z_C^{i,t} \, (\I_C - \frac{1}{N_c} \one_C \one_C^t) \, \Z_C^i, \label{eq-cov-C}
\end{align}
where $\Z_T^i$ and $\Z_C^i$ are the row subsets of included adjustment covariate matrices for treatment and control groups, respectively. $\I_T$ and $\I_C$ are identity matrices of dimensions $N_t$ and $N_c$, and $\one_T$ and $\one_C$ are unit vectors of length $N_t$ and $N_c$, respectively.

TE bias is the first element of the bias vector in Eq.~\ref{eq-bias-vector}, which is the result of multiplying the first row of $\xtxi \, \X^{i,t} \Z^o$ by $\ggamma^o$, while TE variance is simply $\sigma_0^2$ multiplied by the top element of $\xtxi$ (Eq.~\ref{eq-covar-matrix}). Using Equation~\ref{eq-xtxi}, we obtain the following expressions for bias and variance:
\begin{subequations}
\begin{align}
\bias =& \bigg\{ \left( \frac{1}{N_t} + \frac{1}{N_c} + \uu^{i,t} \A^{-1} \uu^i \right) \tr^t + \left( -\frac{1}{N_c} + \frac{1}{N_c} \, (\pp - \qq)^t \A^{-1} \, \uu^i \right) \one^t \nonumber \\
&+ \uu^{i,t} \A^{-1} \Z^{i,t} \bigg\} \, \Z^o \, \ggamma^o. \label{eq-bias-standard} \\
\variance =& \sigma_0^2 \left( \frac{1}{N_t} + \frac{1}{N_c} + \uu^{i,t} \A^{-1} \uu^i \right). \label{eq-variance-standard}
\end{align}
\end{subequations}
We refer to Equations~\ref{eq-bias-standard} and \ref{eq-variance-standard} as the standard equations for bias and variance. As shown in Appendix~\ref{appendix-normalized-eqs}, the above expressions can be transformed into normalized forms:
\begin{subequations}
\begin{align}
\bias &= \left[ \frac{N_t N_c}{(N_t + N_c) (N_t + N_c + 1)} \right]^{1/2} \, \left\{ \rrho^{o,t} + \rrho^{i,t} \wcm^{-1} (\rrho^i \rrho^{o,t} - \PPhi^{io}) \right\} ({\SSigma^{o}}^{1/2} \ggamma^o), \label{eq-bias-normal} \\
\variance &= \sigma_0^2 (\frac{1}{N_t} + \frac{1}{N_c}) (1 + \rrho^{i,t} \, \wcm^{-1} \, \rrho^i) = \sigma_{min}^2 (1 + \rrho^{i,t} \, \wcm^{-1} \, \rrho^i) \label{eq-variance-normal}
\end{align}
\end{subequations}
where $\rrho$ is the vector of treatment-covariate correlations, $\PPhi^{io}$ is the cross-correlation matrix between included and omitted covariates, $\wcm$ is the weighted, pooled correlation matrix, $\SSigma^i$ and $\SSigma^o$ are the diagonal matrix of variances for included and omitted covariates respectively, and $\sigma^2_{min}$ is the minimum achievable variance (Theorem~\ref{theorem-variance-no-replacement}). These symbols are mathematically defined as:
\begin{align}
\rho^i &\equiv [ \, \corr(\tr, \z^i_{,k}) \, ]_{k=1,\hdots,K^i} \\
\rho^o &\equiv [ \, \corr(\tr, \z^o_{,k}) \, ]_{k=1,\hdots,K^o} \\
\PPhi^{io} &\equiv [\, \corr(\z^i_{,k_1}, \z^o_{,k_2}) \, ]_{k_1=1,\hdots,K^i, k_2=1,\hdots,K^o} \\
\wcm &\equiv \SSigma^{i,-1/2} \, \A \, \SSigma^{i,-1/2} / (N_t + N_c - 1) \label{eq-def-wcm} \\
\SSigma^i &\equiv [\, \cov(\z^i_{,k_1}, \z^i_{,k_2}) \, \delta_{k_1,k_2} \,]_{k_1,k_2=1,\hdots,K^i} \label{eq-sigma-i} \\
\SSigma^o &\equiv [\, \cov(\z^o_{,k_1}, \z^o_{,k_2}) \, \delta_{k_1,k_2} \,]_{k_1,k_2=1,\hdots,K^o} \label{eq-sigma-o} \\
\sigma_{min}^2 &\equiv \sigma_0^2 \, (\frac{1}{N_t} + \frac{1}{N_c}) \label{eq-def-sigmin}
\end{align}
The normalized equations make it clear that TE variance is independent of shifting and scaling of covariates. More importantly, the normalized bias expression of Eq.~\ref{eq-bias-normal} makes the role of covariate balance and matching much easier to discern (Theorem~\ref{theorem-matching-bias}).

The vector of mean differences, $\uu^i$, can be normalized to obtain a vector of `standardized mean difference' values for each included covariate, $\dd^i$. These vectors are proportional to the vector of treatment-covariate correlations, $\rrho^i$:
\begin{align}
\dd^i &\equiv - {\SSigma^i}^{-1/2} \, \uu^i, \label{eq-def-dd} \\
\rrho^i &= - \left( \frac{N_t N_c}{(N_t + N_c)(N_t + N_c - 1)} \right)^{1/2} {\SSigma^i}^{-1/2} \, \uu^i, \label{eq-rrho-u} \\
\rrho^i &= \left( \frac{N_t N_c}{(N_t + N_c)(N_t + N_c - 1)} \right)^{1/2} \, \dd^i. \label{eq-rrho-dd}
\end{align}
Eq.~\ref{eq-rrho-u} is proven in Appendix~\ref{appendix-normalized-eqs}, and Eq.~\ref{eq-rrho-dd} follows from Eqs.~\ref{eq-def-dd} and ~\ref{eq-rrho-u}. Similar relationships hold for the corresponding vectors of omitted covariates, $\uu^o$, $\dd^o$ and $\rrho^o$. A corollary of Eq.~\ref{eq-rrho-dd} is the following:

\begin{corollary}
\label{corollary-smd-corr}
If a covariate is balanced across treatment and control groups, it is uncorrelated with the treatment indicator variable.
\end{corollary}
\begin{proof}
Balance for covariate $k$ means its standardized mean difference, $d_k$, is zero. This mean, according to Eq.~\ref{eq-rrho-dd}, that $\rho_k$ is zero.
\qed
\end{proof}

Having derived the standard and normalized expressions for TE bias and variance using linear regression, we are now ready to study the impact of matching on each error component.

\subsection{Impact of matching}\label{subsection-matching-impact}
The impact of matching on TE estimation bias and variance is captured in the following two theorems.

\begin{theorem}
\emph{(Matching and bias)}
\label{theorem-matching-bias}
In a standard linear model, TE estimation is unbiased if all included and omitted covariates have equal means across treatment and control groups.
\end{theorem}
\begin{proof}
The premise means $\dd^i = \dd^o = \zero$. This, according to Corollary~\ref{corollary-smd-corr}, means that $\rrho^i = \rrho^o = \zero$. From Eq.~\ref{eq-bias-normal}, we conclude that $\bias = 0$.
\qed
\end{proof}

While matching reduces TE estimation bias, it does so by discarding data and thus may increase variance. As discussed in the literature~\citep{ho2007matching,stuart2010matching}, matching does not increase TE variance as rapidly as random subsampling of data would. The following theorem formalizes this intuitive and empirical notion.
\begin{theorem}
\emph{(Matching and TE variance)}
\label{theorem-variance-no-replacement}
Among all data sets with a given treatment ($N_t$) and control ($N_c$) group size, generated from a standard linear model with noise variance of $\sigma_0^2$, a data set with balanced distribution of included covariates achieves the lowest (potentially misspecified) OLS-based estimation variance for TE, and this minimum variance equals $\sigma_0^2 \, (1/N_c + 1/N_t)$.
\end{theorem}
\begin{proof}
Since within-group covariance matrices, $\cov(\Z^i)_T$ and $\cov(\Z^i)_C$, are positive semi-definite~\citep{gentle2007matrix}, their weighted sum in Eq.~\ref{eq-def-A} is also positive semi-definite. Since $\A$ is positive semi-definite, $\A^{-1}$ is positive definite (with eigenvalues being inverse of eigenvalues for $\A$). Therefore, by definition of positive-definiteness, $\uu^{i,t} \A^{-1} \uu^i > 0, \,\, \forall \uu^i \neq \zero$. The equality happens when $\uu^i = \zero$, i.e. when included covariates are balanced, leading to $\sigma^2 = \sigma_{min}^2$. As shown in Appendix~\ref{appendix-var-repeat} , sampling with replacement increases variance, and thus the minimum variance established in this theorem applies to matching with replacement as well.
\qed
\end{proof}

A few observations are worth mentioning with regards to the bias/variance expressions and the resulting theorems:
\begin{enumerate}
\item For TE estimation to be unbiased, not only included but also omitted covariates must be balanced. Since we don't know what omitted covariates are (otherwise we would include them in regression and eliminate bias), we generally cannot ensure they are balanced. An exception is perfect matching where, e.g. each treatment observation is perfectly matched - with respect to all original covariates - with a control observation. But thanks to ignorability assumption, we are ensured that all omitted covariates are also functions of included covariates. Therefore, perfect matching on included covariates would automatically lead to perfect matching on omitted covariates. But perfect matching is often impractical, either because some covariates are continuous, or because perfect matching leads to an excessively small sample size and hence large variance, or both. \cite{rubin1973use} similarly observe that matching on covariates in a linear regression or ANCOVA analysis does not guarantee bias removal for nonlinear response surfaces.
\item In the absence of perfect matching, we must accept the fact that TE will be biased, and instead do our best to minimize it. Again, our main challenge is that, according to Eq.~\ref{eq-bias-normal}, bias is a function of included as well as omitted covariates. This represents a key difference between bias and variance equations, since variance, according to Eq.~\ref{eq-variance-normal} only depends on included covariates. In other words, quantifying bias is a more difficult proposition than quantifying variance. In Section~\ref{subsection-quantify-bias}, we will present a framework for quantifying bias that is independent of outcome variable.
\item In Theorem~\ref{theorem-variance-no-replacement}, we showed that a data set that is balanced with respect to included covariates achieves $\sigma_{min}^2$. A second way to achieve minimum variance is to simply include no adjustment covariates in the regression, i.e. only have intercept and treatment indicator. This is the simple difference method mentioned in Section~\ref{subsection-beyond-double}. As we saw in Figure~\ref{fig-double-robustness}, while this approach has low variance, using it with an unmatched data often produces a large enough bias that more than offsets the low variance and leads to large MSE.
\end{enumerate}

\subsection{TE bias and orthogonality}\label{subsection-bias-orthogonality}
Neither the standard (Eqs.~\ref{eq-bias-standard} and \ref{eq-variance-standard}) nor the normalized (Eq.~\ref{eq-bias-normal} and \ref{eq-variance-normal}) expressions for bias make an intuitive concept clear: An omitted covariate that is very `similar' to the included covariates will not induce much bias, since its contribution to the outcome will be captured by the coefficient of the included covariate acting as a surrogate for the coefficient of omitted covariate. This concept can be formally stated in the following theorem:

\begin{theorem}
\emph{(Bias and orthogonality)}
\label{theorem-bias-orthogonality}
Projections of omitted covariates onto the subspace spanned by included adjustment covariates and intercept variable produce zero TE bias, i.e. only components of omitted covariates orthogonal to the aforementioned subspace contribute towards bias.
\end{theorem}
\begin{proof}
Equation~\ref{eq-bias-vector} can be expanded in terms of columns of $\Z$:
\begin{equation}
\E[\hbeta] - \bbeta  = \xtxi \X^{t} \sum_{k=1}^{K^o} \z^o_{,k} \gamma_k.
\end{equation}
To prove our theorem, and recalling our convention for arranging covariates in $\X^i$ according to Eq.~\ref{eq-def-Xi}, we simply prove that the matrix operator $\xtxi \X^{i,t}$ acting on any one of $K$ columns of $\Z^i$, i.e. $\z^i_{,k}$, produces a $K$-dimensional vector whose first element is zero. This is easy to see since, by definition, $\xtxi \xtx = \I_{K^i+2}$, where $\I_{K^i+2}$ is the identity matrix of dimensions $K^i+2$. Therfore, $\xtxi \X^{i,t} \z^i_{,k}$ equals the $(k+2)$'th column of the identity matrix $\I_{K+2}$, whose first two elements will always be zero. Therefore, omission of any linear combination of $\z^i_{,k}$'s will produce zero TE bias.
\qed
\end{proof}
This theorem will be used in Section~\ref{subsection-quantify-bias} to develop the constrained bias estimation approach.

\section{Applications}\label{section-tools}
In Section~\ref{section-framework}, we developed the mathematical framework for quantifying TE bias and variance, as well as the impact of matching, in a standard linear model (Section~\ref{subsection-problem-definition}). In addition to providing a theoretical basis for understanding the combined use of matching and linear regression, this framework can be utilized towards developing diagnostic and calibration tools for causal inference. Such applications are the subject of this section.

\subsection{Efficiency and generalizability of simulations}\label{subsection-efficient-simulation}
The first application of the framework developed in Section~\ref{section-framework}, particularly the closed-form expressions for bias and variance (Eqs.~\ref{eq-bias-standard}, \ref{eq-variance-standard}, \ref{eq-bias-normal}, and \ref{eq-variance-normal}), is that it allows for fast and accurate simulations of the combined effect of matching and variance without requiring lengthy Monte Carlo simulations. Figure~\ref{fig-verify-equations} compares the bias and variance calculations using our closed-form expressions - encoded in the R package MatchLinReg (see Section~\ref{subsection-software}) - as well as Monte Carlo simulations, using 10,000 iterations to generate each point. While closed-form expressions produce results nearly instantaneously, the MC approach took several minutes on an average laptop, despite the data sets being quite small (fewer than 1000 observations). For larger data sets, the advantage of our framework becomes more prominent. Such efficient simulation tool allows for exhaustive exploration of data sets and parameter spaces in future research, ultimately leading to better empirical rules for combining matching and linear regression.

\begin{figure}
\centering
\begin{tabular}{cc}
\subfloat{\includegraphics[height=3.2in]{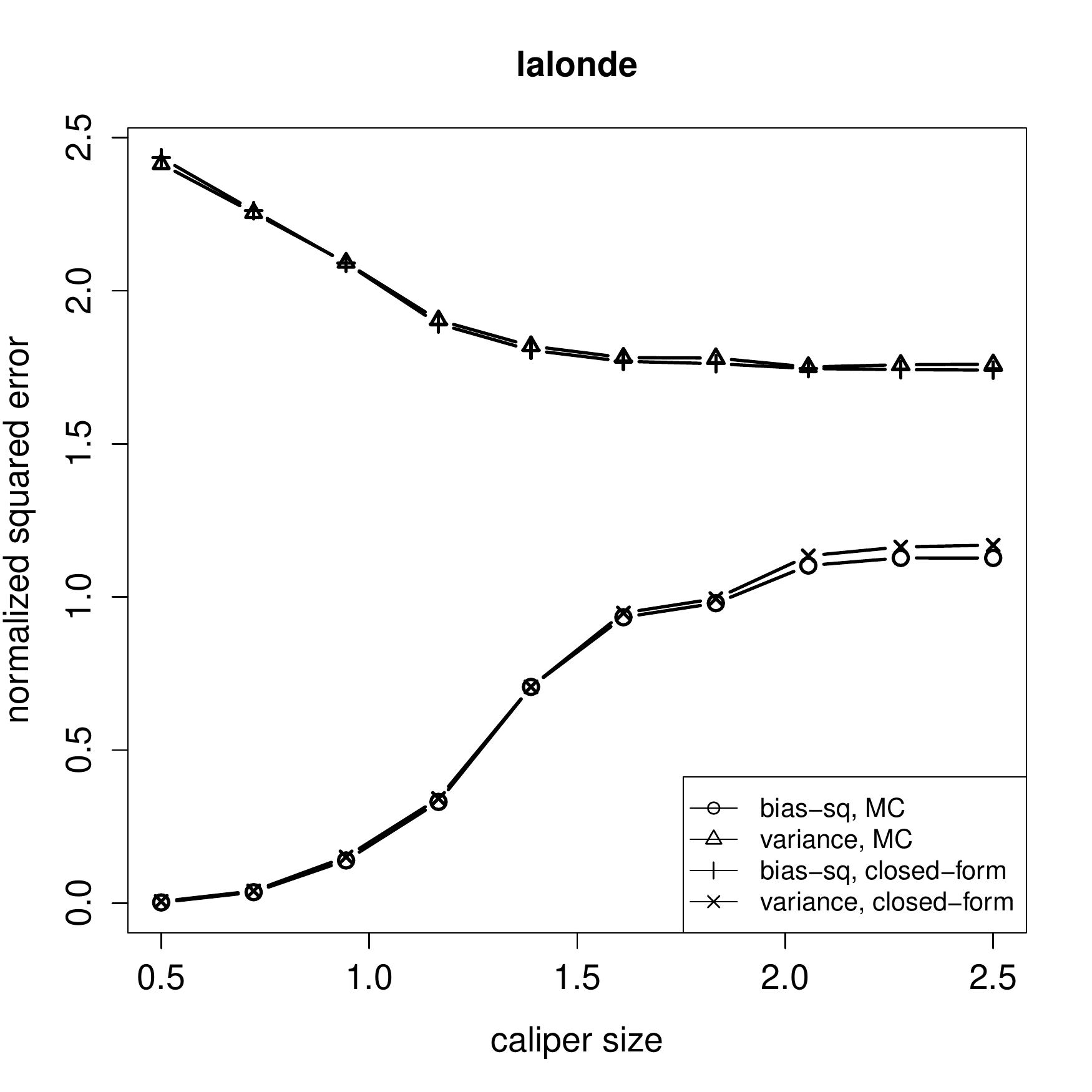}} &
\subfloat{\includegraphics[height=3.2in]{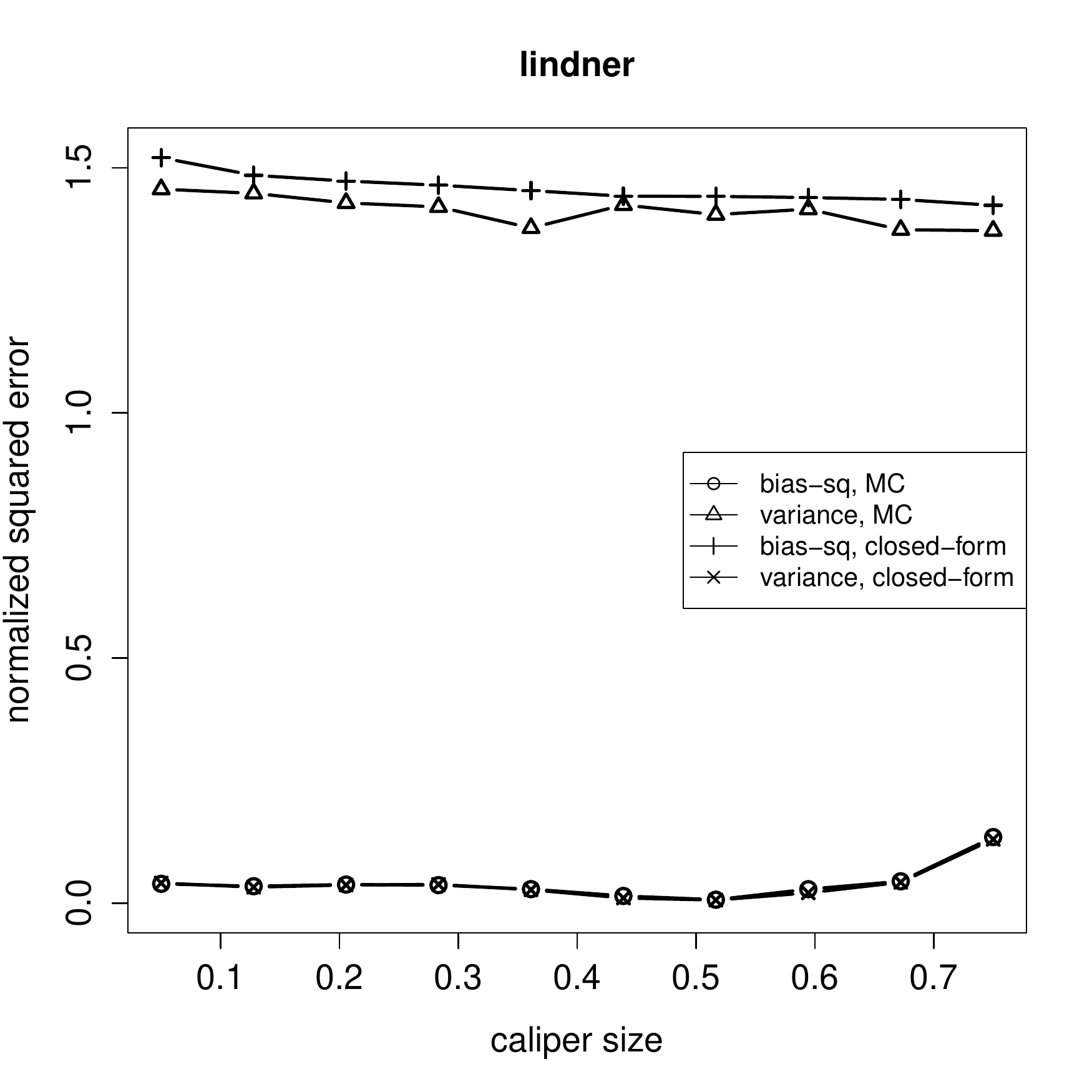}} 
\end{tabular}
\caption{Comparison of TE bias and variance calculations using Monte Carlo simulations (10,000 iterations per data point) and closed-form expressions of our framework for lalonde (left) and lindner (right) data sets. MC calculations took $\sim$10-15min, while closed-form calculations took less than 0.1sec. Data generation parameters are the same as those reported for Figure~\ref{fig-double-robustness}.}
\label{fig-verify-equations}
\end{figure}

In addition to speed of simulations, an equally-important contribution of our framework is that it provides for better generalizability of simulation results since we have a theoretical basis for what are the drivers of each error component. For example, Eqs.~\ref{eq-bias-standard} and \ref{eq-variance-standard} make it clear that TE bias and variance do not depend on coefficients of adjustment covariates included in the regression model. Therefore, this aspect of data sets can be safely ignored while comparing their simulation results across studies. The enhanced generalizability of simulations resulting from utilizing our framework directly addresses the requirement stated in \cite{imbens2004nonparametric}. This paper is the first beneficiary of the framework, as all the simulations in the rest of this paper utilize the aforementioned closed-form expressions.

\subsection{Quantifying bias}\label{subsection-quantify-bias}
A key challenge in assessing the benefits of matching is that TE bias is a function of the unknown, omitted covariates and their coefficients (Eqs.~\ref{eq-bias-standard} and \ref{eq-bias-normal}) and hence any bias-removal impact of matching is similarly unknown, and hard to quantify. Naively, it might seem like a reasonable idea to use `data exploration' to identify such omitted covariates and the strength of their impact, e.g. using pairwise correlation analysis of various candidate terms with outcome. Such explorations, however, can lead to overfitted models, particularly in the hands of inexperienced practitioners~\citep{babyak2004you,ho2007matching}. Furthermore, if we could somehow learn what the omitted covariates are, we could simply include them in the regression model and eliminate the bias that way, rather than using matching. Instead, we opt for a dual approach:
\begin{enumerate}
\item Bias-oriented analysis: We calculate maximum TE bias, normalized by the upper bound on the amount of unexplained variation in outcome that is due to omitted covariates, and where the vector of contributions from omitted covariates lies within a pre-specified subspace. This normalized bias is then combined with TE variance over a range of values for the bias-to-variance conversion factor, which we call `omitted R-squared'. This analysis begins with a set of diagnostic tools (Sections~\ref{subsubsection-percent-bias-reduction}, \ref{subsubsection-compare-normalized-biases} and \ref{subsubsection-constrained-bias-estimation}), and culminates in a calibration tool (Section~\ref{subsection-combine-bias-var}). As it emphasizes bias, this analysis favors matching.
\item Variance-oriented analysis: We assume the other extreme with regards to the strength of omitted covariates, namely that it is zero. Assuming a correct model specification - and hence zero bias - we calculate the reduced study power resulting from the choice of matching parameters in previous analysis (Section~\ref{subsection-power-analysis}). This represents a worst-case-scenario as far as the negative impact of matching on TE estimation error.
\end{enumerate}
Practitioners must weight the benefits of matching (from the first analysis) against the worst-case-scenario from the second analysis to determine the best course of action. An interesting area for future research is to combine these two facets into a unified framework, possibly using risk-based Bayesian methods to offer an even more prescriptive path towards matching calibration for practitioners.

With the above roadmap in mind, we begin developing a methodology for quantifying TE bias in misspecified regression. Eq.~\ref{eq-bias-standard} can be re-written as follows:
\begin{align}
\bias = \g^t \, \Z^o \ggamma^o = \sum_{k=1}^{K^o} \g^t (\gamma_k^o \, \z^o_{,k}), \label{eq-bias-sum}
\end{align}
where
\begin{align}
\g \equiv \left( \frac{1}{N_t} + \frac{1}{N_c} + \uu^{i,t} \A^{-1} \uu^i \right) \tr + \left( -\frac{1}{N_c} + \frac{1}{N_c} \, (\pp - \qq)^t \A^{-1} \, \uu^i \right) \one + \Z^i \A^{-1} \uu^i,
\end{align}
and $\z^o_{,k}$ is the $k$'th omitted covariate, $\gamma_k^o$ is its coefficient, and $\g$ is a vector of length $N$. Eq.~\ref{eq-bias-sum} means that the contribution of each omitted covariate towards TE bias is 1) additive, 2) proportional to its coefficient, and 3) dependent on how parallel it is with $\g$.

While we do not know exactly what the omitted covariates are (otherwise they would have been included), we can often form a reasonable candidate set, e.g. by forming interactions (second-order or higher) and powers of the included covariates. Given the additive property of bias, studying the impact of matching on bias due to each potentially-omitted covariate is a meaningful first step.

\subsubsection{Single-covariate relative squared bias reduction}\label{subsubsection-percent-bias-reduction}
A first diagnostic tool is to compare (squared) bias caused by a candidate omitted covariate, before and after matching. A relative bias reduction metric will be independent of $\gamma^o_k$:
\begin{align}
\frac{\delta_{\tau,k}^{2,i} - \delta_{\tau,k}^{2,f}}{\delta_{\tau,k}^{2,i}} = \frac{|| \g^t \, \z^{o,i}_{,k} ||^2 - || \g^t \, \z^{o,f}_{,k} ||^2}{|| \g^t \, \z^{o,i}_{,k} ||^2}
\end{align}
where $i$ and $f$ superscripts refer to before and after matching, respectively. This quantity (or its square) can be produced for all candidate omitted terms, producing plots similar to those shown in Figure~\ref{fig-percent-bias-reduction}. Ideally, we would like matching to eliminate, or at least reduce, bias for all terms, but the figure shows that this may not happen. Practitioners may need to examine terms for which bias increases more closely, including performing the next analysis.

\begin{figure}
\centering
\begin{tabular}{cc}
\subfloat{\includegraphics[height=3.2in]{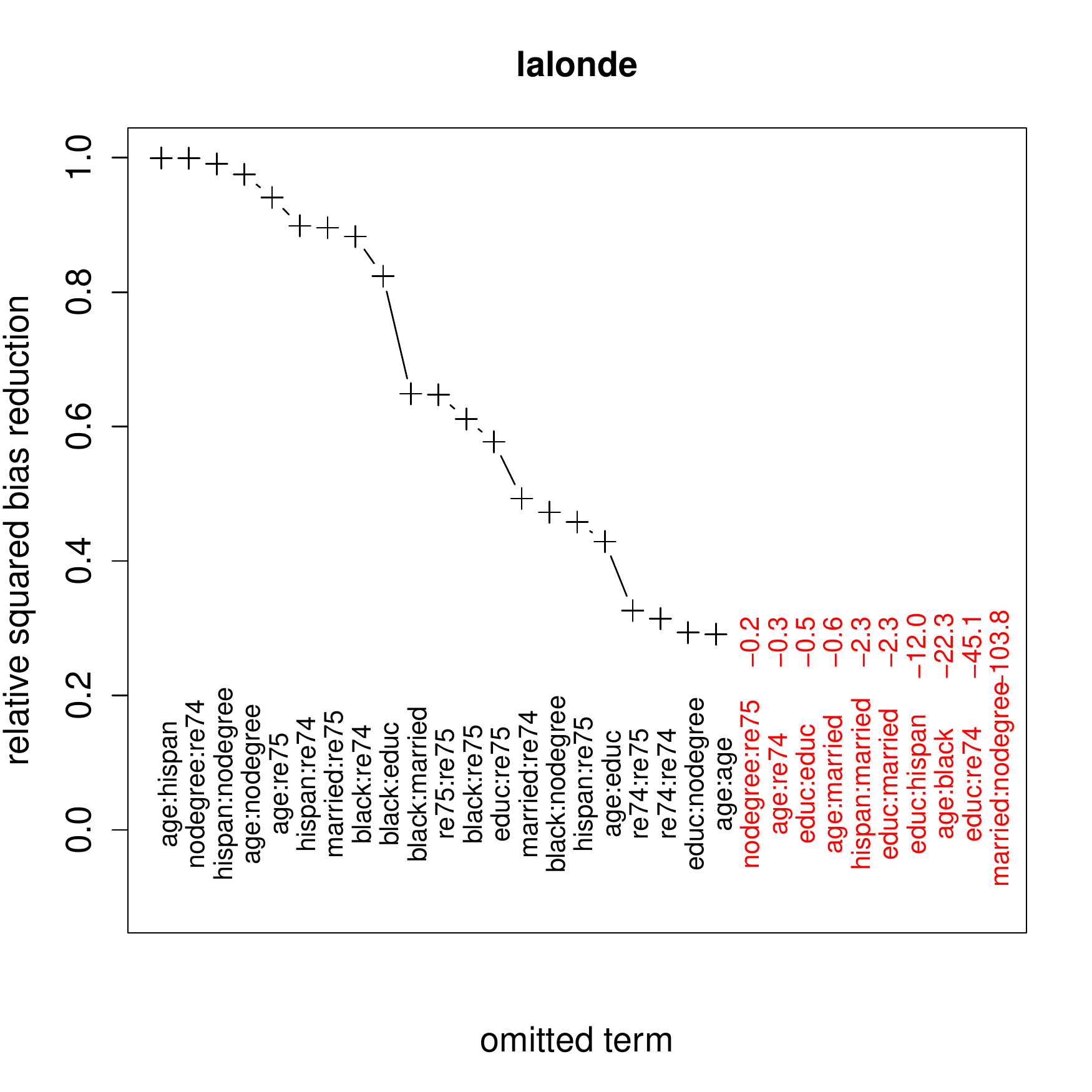}} & 
\subfloat{\includegraphics[height=3.2in]{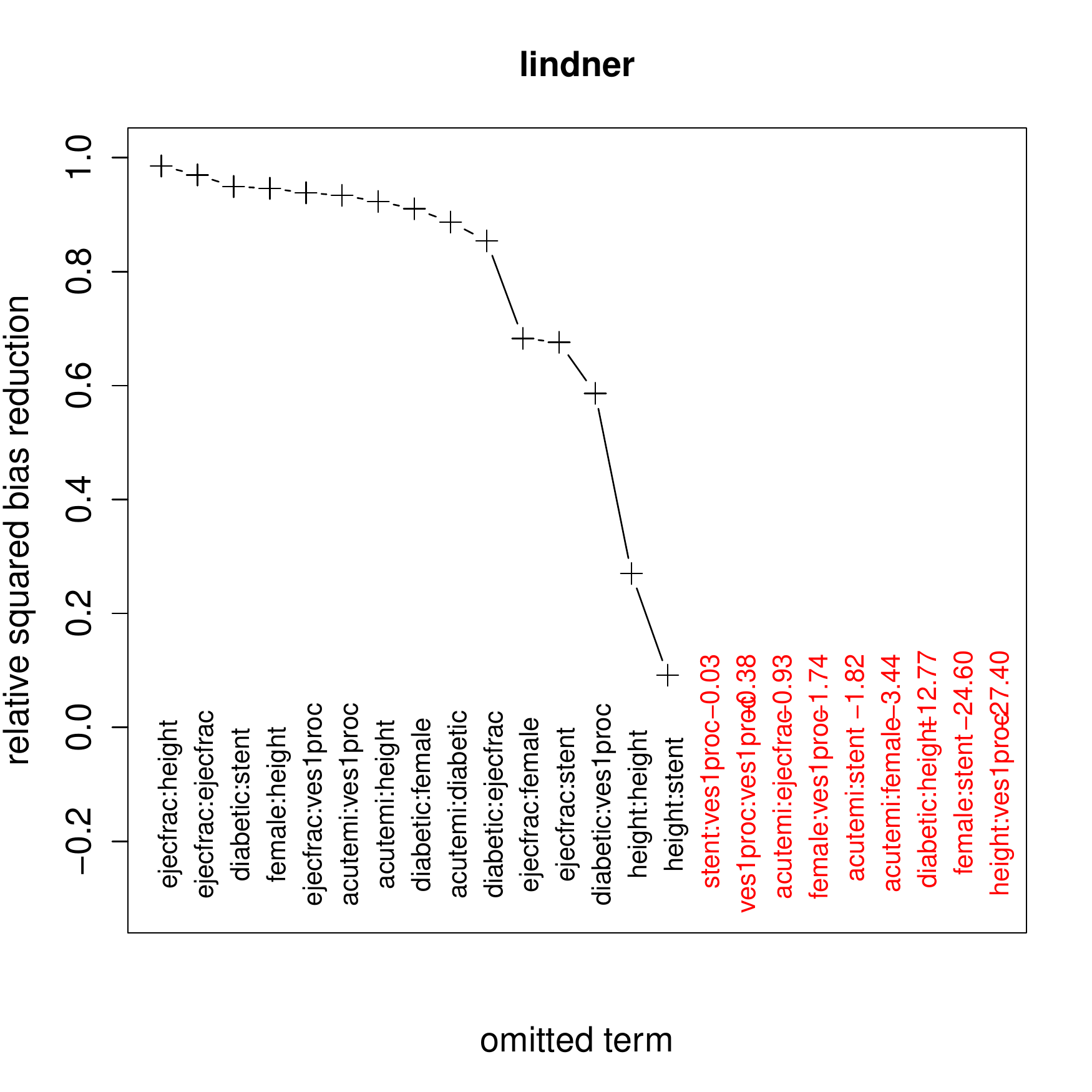}} \\ 
\end{tabular}
\caption{Relative squared bias reduction for all second-order terms as candidate omitted covariates. Results for lalonde (left) and lindner (right) data sets are shown. Regression adjustment using all main effects (`before') is supplemented by propensity score matching as a pre-processing step (`after'), using caliper size of 1.5 for lalonde and 0.2 for lindner. Omitted terms are sorted by decreasing relative bias reduction. Terms in red are those for which matching has increased squared (or absolute) bias.}
\label{fig-percent-bias-reduction}
\end{figure}

\subsubsection{Comparison of single-covariate normalized biases} \label{subsubsection-compare-normalized-biases}
While relative bias reduction is calculated independently of omitted coefficients ($\ggamma^o$), yet it does not provide any comparison of the magnitude of biases induced by each omitted covariate. Matching may not significantly reduce the bias due to an omitted covariate, yet if the overall bias contributed by that covariate is likely to be small, lack of bias reduction by matching for that covariate becomes unimportant.

Since each term's contribution is proportional to its coefficient ($\gamma^o_k$), we must somehow determine these coefficients. As emphasized before, we opt for an approach that does not use the outcome variable. If omitted covariates have been orthogonalized with respect to $\{ \one, \Z^i \}$ (the subspace of included adjustment covariates plus intercept), their mean squared contribution to outcome is roughly the same as the change in R-squared of the regression model, multiplied by $\sigma_0^2$. This motivates us to define a parameter called `omitted R-squared': it is the ratio of mean squared contribution to outcome for an omitted covariate (or a set of covariates), after orthogonalization, to noise variance:
\begin{align}
R_o^2 = || \Z_{\bot}^o \ggamma^o ||^2 / \sigma_0^2
\end{align}
This parameter can be intuitively described as the percentage of variation in outcome that is not explained by the regression model. If included covariates are all linear terms (in original covariates), then one can also think of $R_o^2$ as the degree of nonlinearity in the data-generation process. Note that we do not propose that this parameter be extracted from the data, e.g. through the use of outcome variable. Rather, this parameter must be determined based on a domain expert's knowledge of the underlying mechanisms involved as well as experience and future simulation studies. We will use this parameter in Section~\ref{subsection-combine-bias-var} to put bias on the same scale as variance and combine them to arrive at total MSE for TE.

Given the above, we propose a constrained bias estimation approach where we estimate bias subject to the constraint that mean squared contribution from omitted covariates equals a given value. In fact, we can canonically set this value to 1 and estimate a normalized bias value:
\begin{align}
\bias^n = \sqrt{N} \, \g^t \, \Z_{\bot}^o \ggamma^o \, / \, || \Z_{\bot}^o \ggamma^o || \label{eq-normalized-bias}
\end{align}
Normalized bias has an intuitive interpretation: it is the fraction of unexplained/omitted variation in outcome (or `signal') that turns into TE bias. For a perfectly-matched data set, this fraction is zero.

Within this general approach, several variations can be conceived of. We begin with the simplest approach, where we calculate and compare normalized bias for each of the candidate omitted covariates. In other words, we permit $\ggamma^o$ to have exactly one non-zero element. This produces $K^o$ normalized bias numbers, $\delta_k^n$:
\begin{align}
\delta_k^n &= \sqrt{N} \, \g^t \, \gamma^o_k \z^o_{\bot,k} \, / \, || \gamma^o_k \z^o_{\bot,k} ||, \label{eq-normalized-bias} \nonumber \\
&= \pm \, \sqrt{N} \, \g^t \, \z^o_{\bot,k} / || \z^o_{\bot,k} ||.
\end{align}
We can now produce and compare these single-covariate, normalized biases for a set of candidate omitted terms, perhaps before and after matching, and for different calibrations of matching to identify nonlinearities that can potentially induce large residual bias, even after matching. Since absolute values are more important than signs, we focus on squared bias, a metric which is also dimensionally compatible with variance and total MSE.

An example is provided in Figure~\ref{fig-compare-normalized-biases}, where single-covariate normalized square bias has been plotted for a handful of second-order interactions as a function of matching caliper size. We see that matching does manage to reduce the normalized bias for the biggest offenders. A notable exception is the \textit{female:stent} interaction in the lindner data set, where matching has raised the bias to relatively significant level after matching ($\sim 1\%$ of omitted variance). Also, note that the overall normalized bias levels are significantly higher in lalonde data set, compared to lindner. In other words, a larger fraction of omitted signal translates into TE bias in lalonde data set. Figure~\ref{fig-compare-normalized-biases} also illustrates that impact of matching on covariate imbalance (measured via absolute mean difference) does not have a simple relationship with bias reduction.

\begin{figure}
\centering
\begin{tabular}{cc}
\subfloat{\includegraphics[height=3.2in]{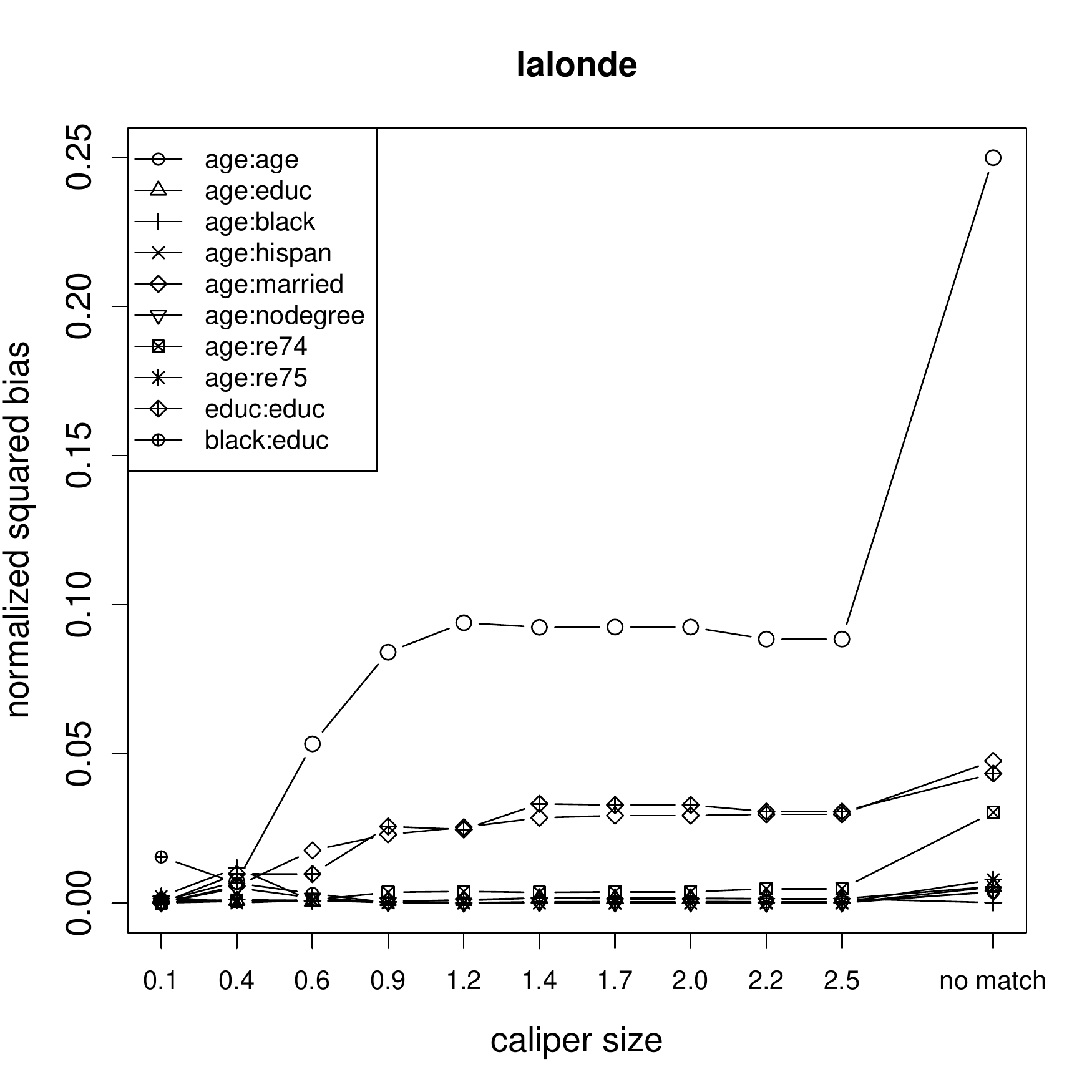}} & 
\subfloat{\includegraphics[height=3.2in]{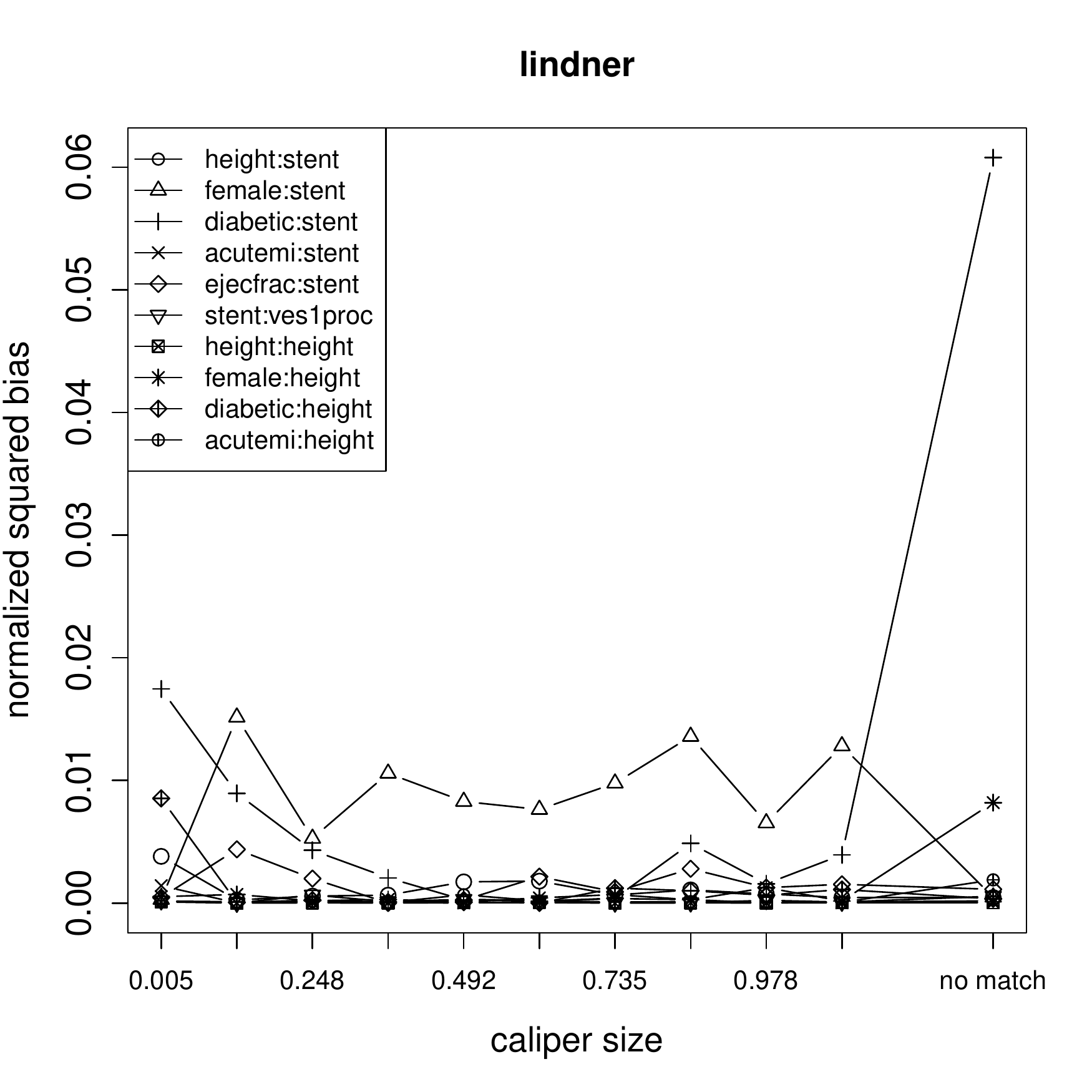}} \\ 
\subfloat{\includegraphics[height=3.2in]{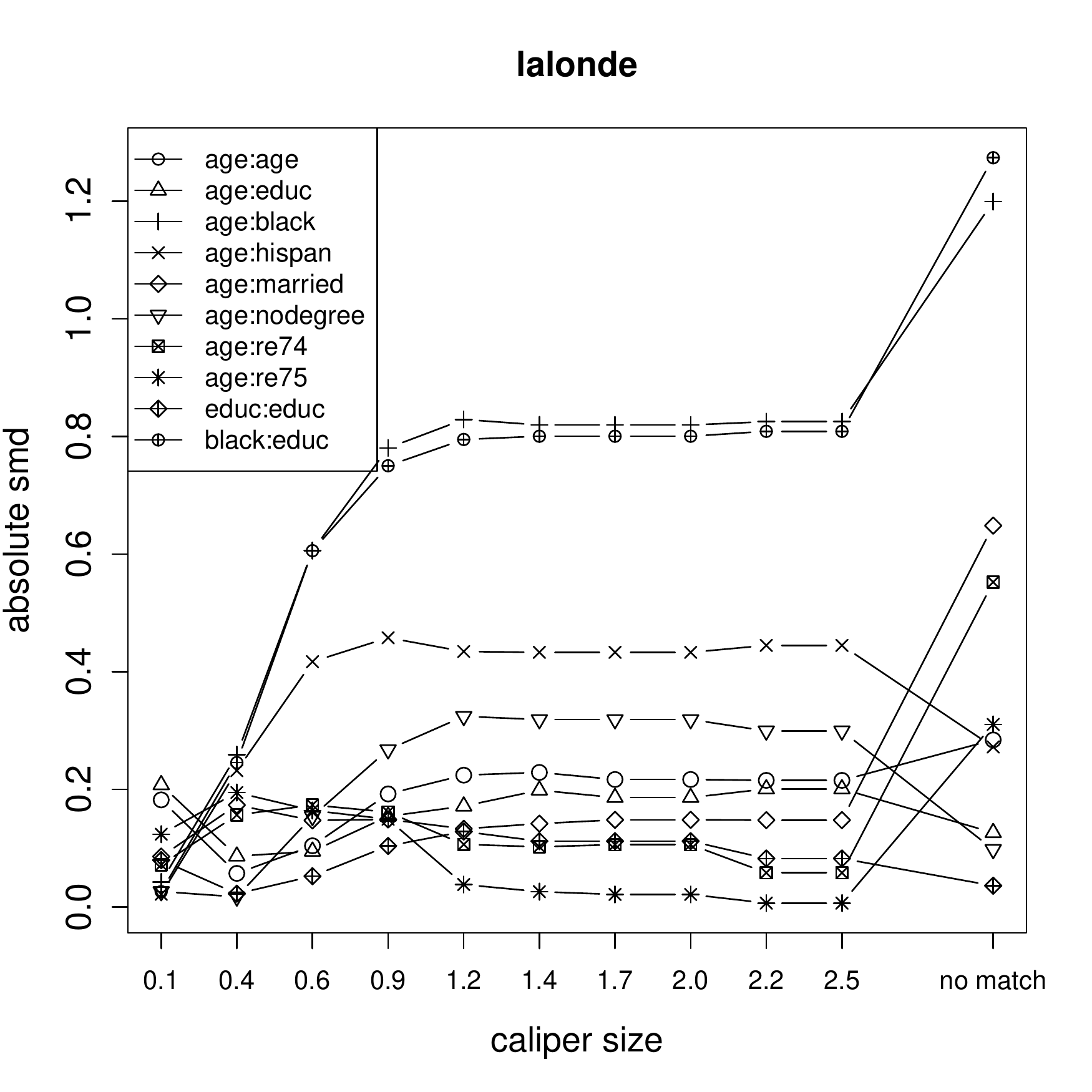}} & 
\subfloat{\includegraphics[height=3.2in]{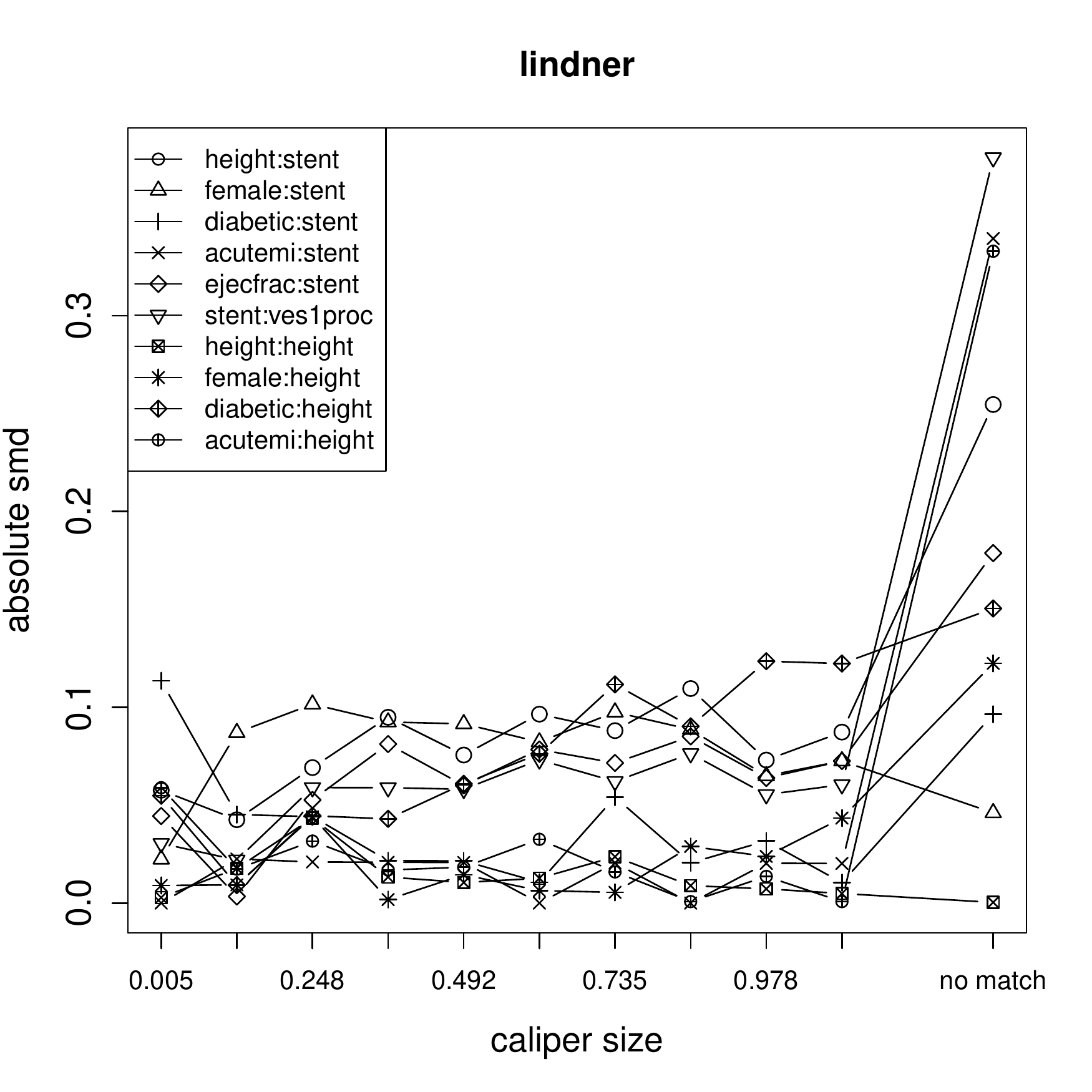}} \\ 
\end{tabular}
\caption{Top row: Comparison of normalized squared biases due to omission of 10 second-order terms from regression adjustment for lalonde (left) and lindner (right) data sets. Horizontal axis is caliper size used in propensity score matching, with the right-most point on each plot corresponding to no matching. Bottom row: Impact of matching on absolute mean differences for same 10 interaction terms.}
\label{fig-compare-normalized-biases}
\end{figure}

\subsubsection{Estimating aggregate bias}\label{subsubsection-constrained-bias-estimation}
While producing and examining normalized squared biases for interaction terms is insightful, it stops short of producing a single, aggregate measure of bias, which we would need to combine with variance to arrive at an estimated MSE. We consider three aggregation methods below. They all produce a normalized bias estimate (Eq.~\ref{eq-normalized-bias}), and are all based on maximizing normalized bias within a given, eligible subspace. They differ in how this eligible subspace is constructed.

\textit{Single-covariate maximization:} Eligible subspace is one of $K^o$ candidate omitted covariates. The aggregate bias estimate is simply the maximum value of normalized bias terms due to each candidate, omitted covariate, $\{ \delta_{\tau,k}^n \}$.

\textit{Covariate-subspace maximization:} Here the eligible subspace is the entire covariate subspace that is orthogonal to $\{ \one, \Z^i \}$, and the aggregate bias estimate is the linear combination of omitted covariates - after orthogonalization - that maximizes normalized bias. This bias estimate is bound to be greater than or equal to the output from single-covariate maximization approach. The direction of $\Z^o \ggamma^o$ must be parallel to the projection of $\g$ in the omitted covariate subspace, after orthogonalization ($\g_{\parallel}$):
\begin{align}
\bias = \sqrt{N} \, \g^t \g_{\parallel} / (\g_{\parallel}^t \g_{\parallel})^{1/2}
\end{align}

\textit{Absolute maximization:} The eligible subspace is the entire subspace orthogonal to $\{ \one, \X^i \}$, and the aggregate bias estimate in this case is the absolute maximum normalized bias achievable. Here the direction of $\Z^o \ggamma^o$ is simply parallel to $\g$:
\begin{align}
\bias &= \g^t (\frac{\g}{||\g||} \sqrt{N}) = \sqrt{N} || \g ||, \\
&= \sqrt{N} \, \frac{\sigma}{\sigma_0}.
\end{align}

Figure~\ref{fig-compare-bias-methods} compares the bias estimated for these three methods, along with impact of matching on bias reduction for each method. Of these three measures, we believe the covariate-subspace approach is the most suitable approach. While absolute maximization is too extreme as far as bearing no connection to adjustment covariates, single-covariate maximization approach is too rigid in focusing on a very limited set of possible combinations of candidate omitted covariates. Interestingly, bias estimate using absolute maximization does not change with matching. This property is proven in Appendix~\ref{appendix-absolute-matching}.

\begin{figure}
\centering
\begin{tabular}{cc}
\subfloat{\includegraphics[height=3.2in]{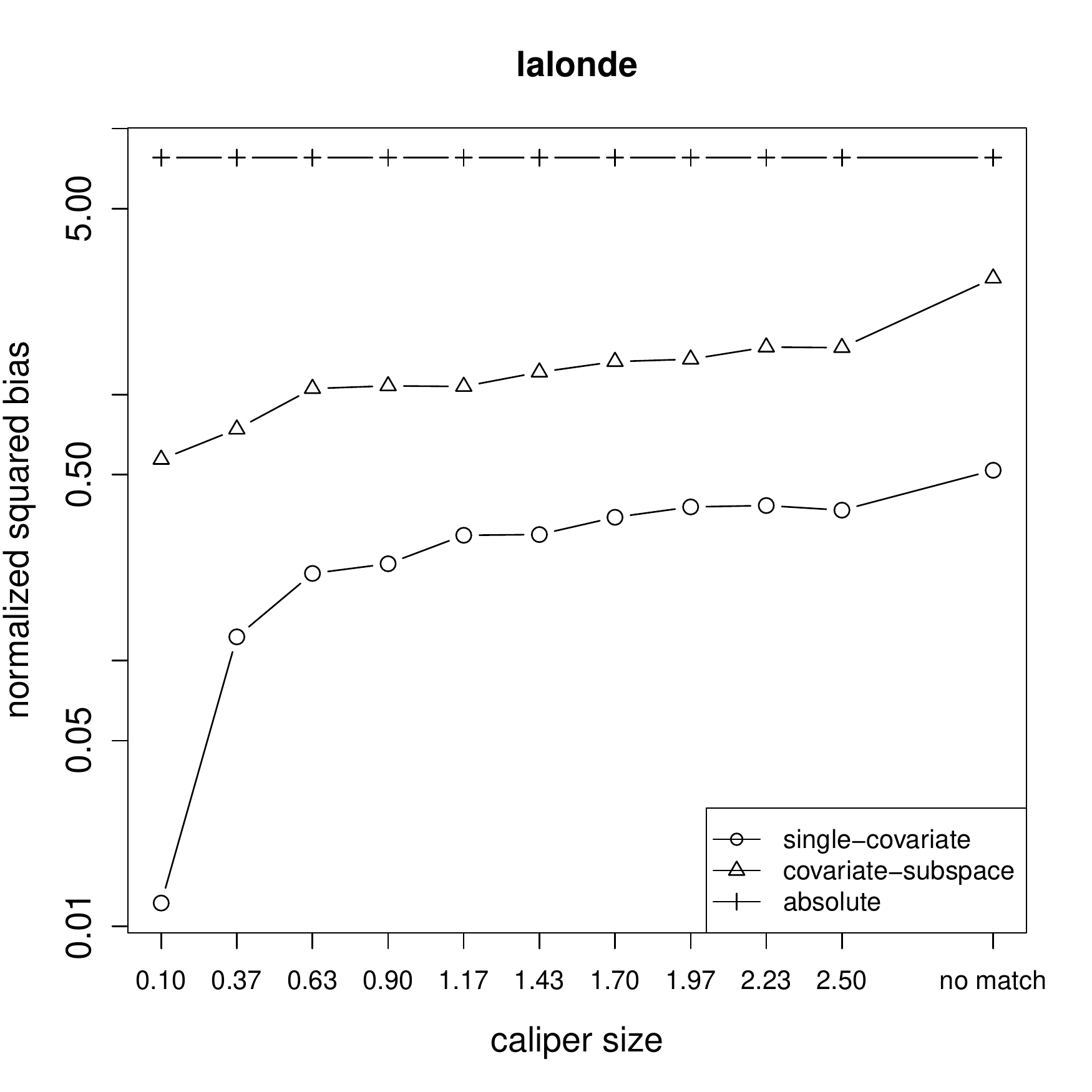}} & 
\subfloat{\includegraphics[height=3.2in]{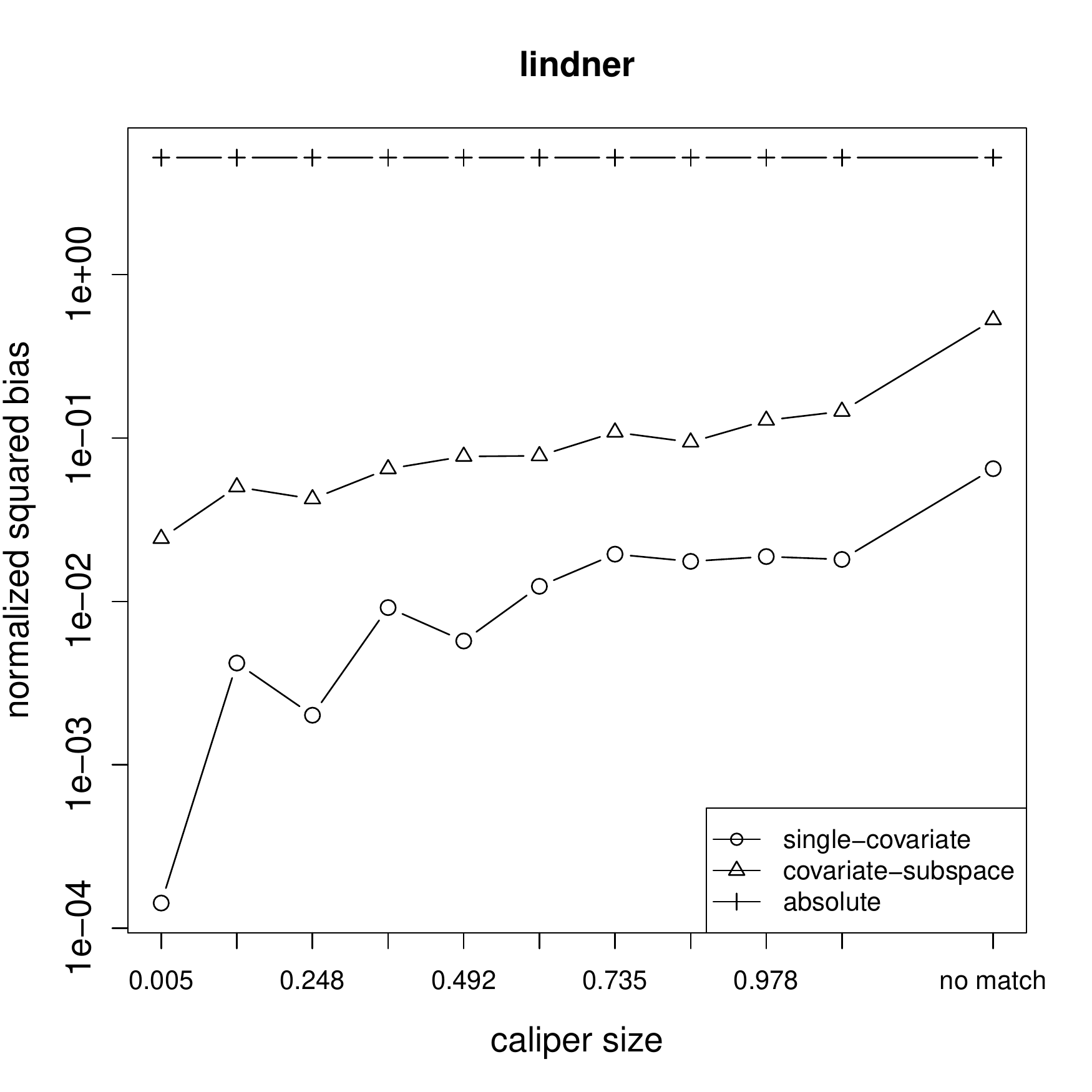}}
\end{tabular}
\caption{Impact of matching and caliper size on TE bias using single-covariate maximization, covariate-space maximization, and absolute maximization for lalonde (left) and lindner (right) data sets.}
\label{fig-compare-bias-methods}
\end{figure}

It is possible to replace maximization with another operator - such as averaging - in each of the above methods. The fact that using the maximum operator tends to exaggerate the impact of bias motivates us to present a counter-view in which we assess the damage caused by matching towards variance increase and the resulting loss of study power, assuming a total absence of bias. This is described in Section~\ref{subsection-power-analysis}.

\subsection{Combining bias and variance}\label{subsection-combine-bias-var}

In Section~\ref{subsection-quantify-bias} we developed the constrained bias estimation method for arriving at `normalized squared bias', i.e. the ratio of squared bias to omitted signal. On the other hand, Eqs.~\ref{eq-variance-standard} and \ref{eq-variance-normal} allow us to calculate `normalized variance', i.e. the ratio of TE variance to $\sigma_0^2$, in a straightforward manner. In order to combine normalized bias and variance to arrive at MSE, we must have an `exchange rate'. We propose an intuitive parameter, `omitted R-squared', which we simply define to be the ratio of omitted signal to generative noise. If we define `normalized MSE' as the ratio of MSE to $\sigma_0^2$, we have:
\begin{equation}
\mathrm{normalized\,\,MSE} = \mathrm{normalized\,\,variance} + R_o^2 \times \mathrm{normalized\,\,squared\,\,bias}
\end{equation}
Figure~\ref{fig-combine-bias-variance} shows normalized MSE as a function of matching caliper size for various values of $R_o^2$, for lalonde and lindner data sets. For small values of $R_o^2$, variance dominates bias and thus optimal caliper size (that minimized MSE) tends to be large (or we may choose no caliper, or no matching) (top row). As $R_o^2$ gets larger, bias and variance become comparable and thus the bias-removal impact of matching outweighs its variance increase; thus optimal caliper size shifts towards smaller values (second row). For even larger $R_o^2$ (third row), bias dominates variance and even smaller caliper sizes become optimal. This trend is seen more clearly in the bottom row, where optimal caliper size has been plotted as a function of $R_o^2$. These plots can be considered the most prescriptive of all our analyses so far: they suggest a particular calibration (caliper size here) assuming different values of omitted R-squared. Same type of analysis can be applied to other aspects of matching, e.g. what terms to include in matching, propensity score matching vs. Mahalanobis matching, etc.

\begin{figure}
\centering
\begin{tabular}{cc}
\subfloat{\includegraphics[height=2.in]{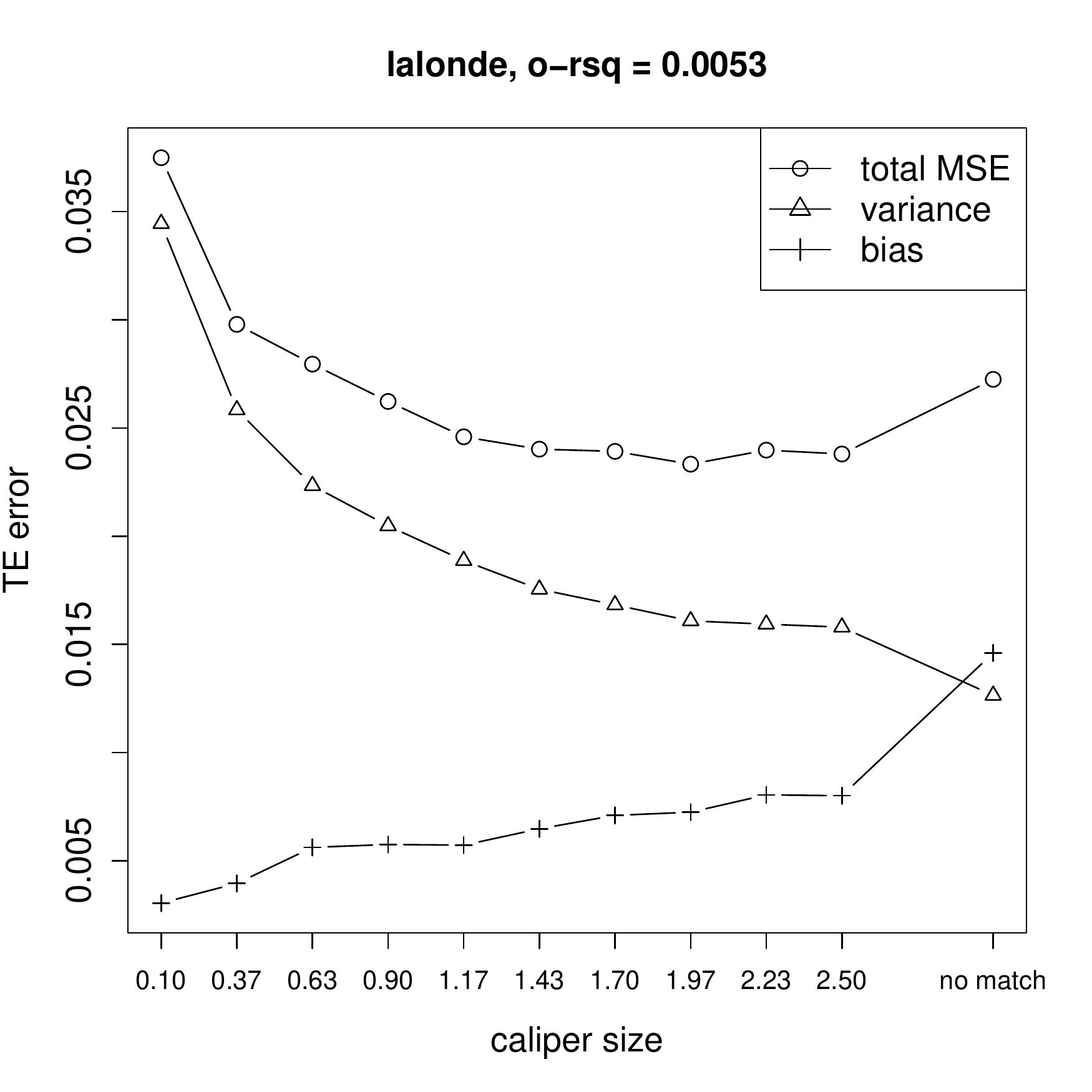}} & 
\subfloat{\includegraphics[height=2.in]{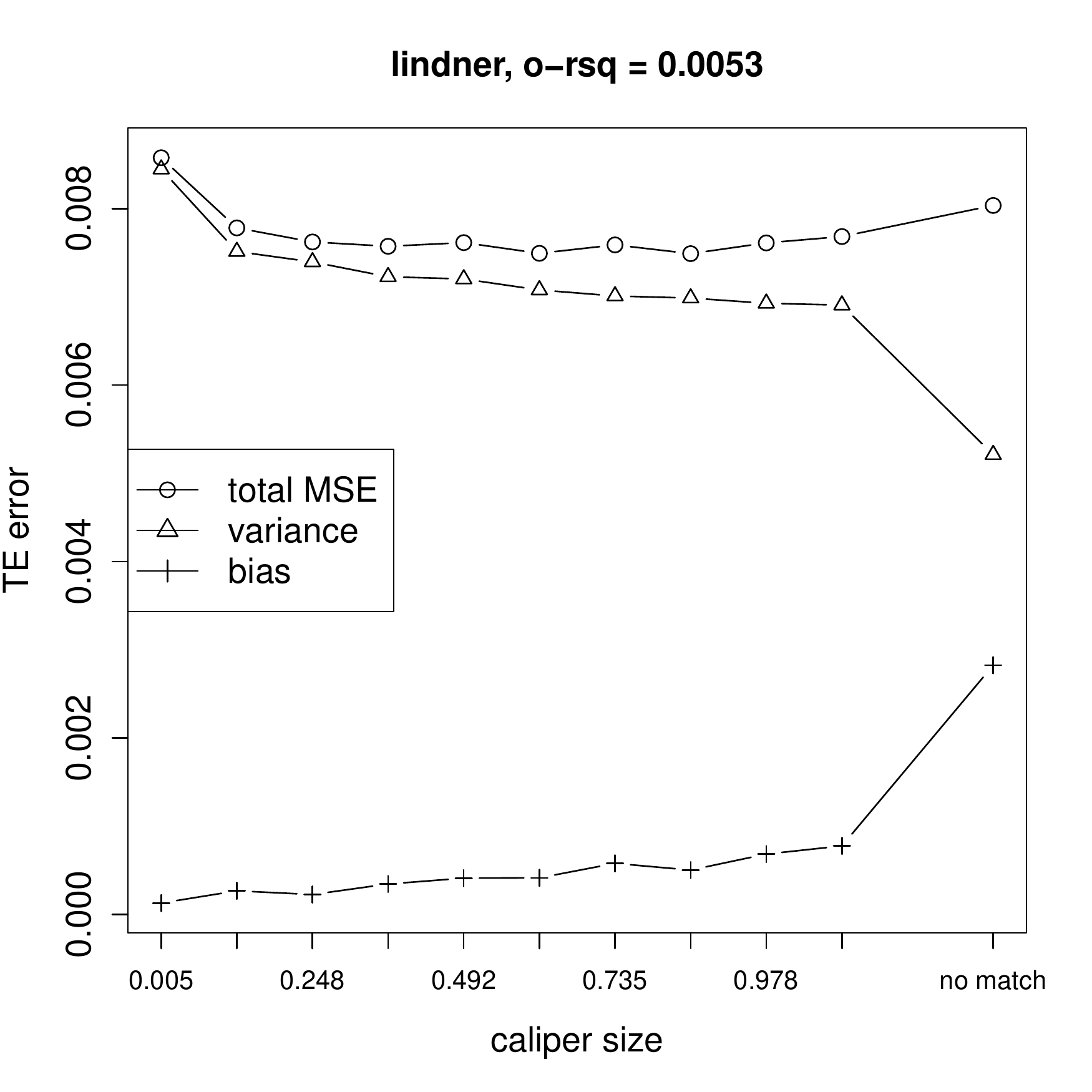}} \\
\subfloat{\includegraphics[height=2.in]{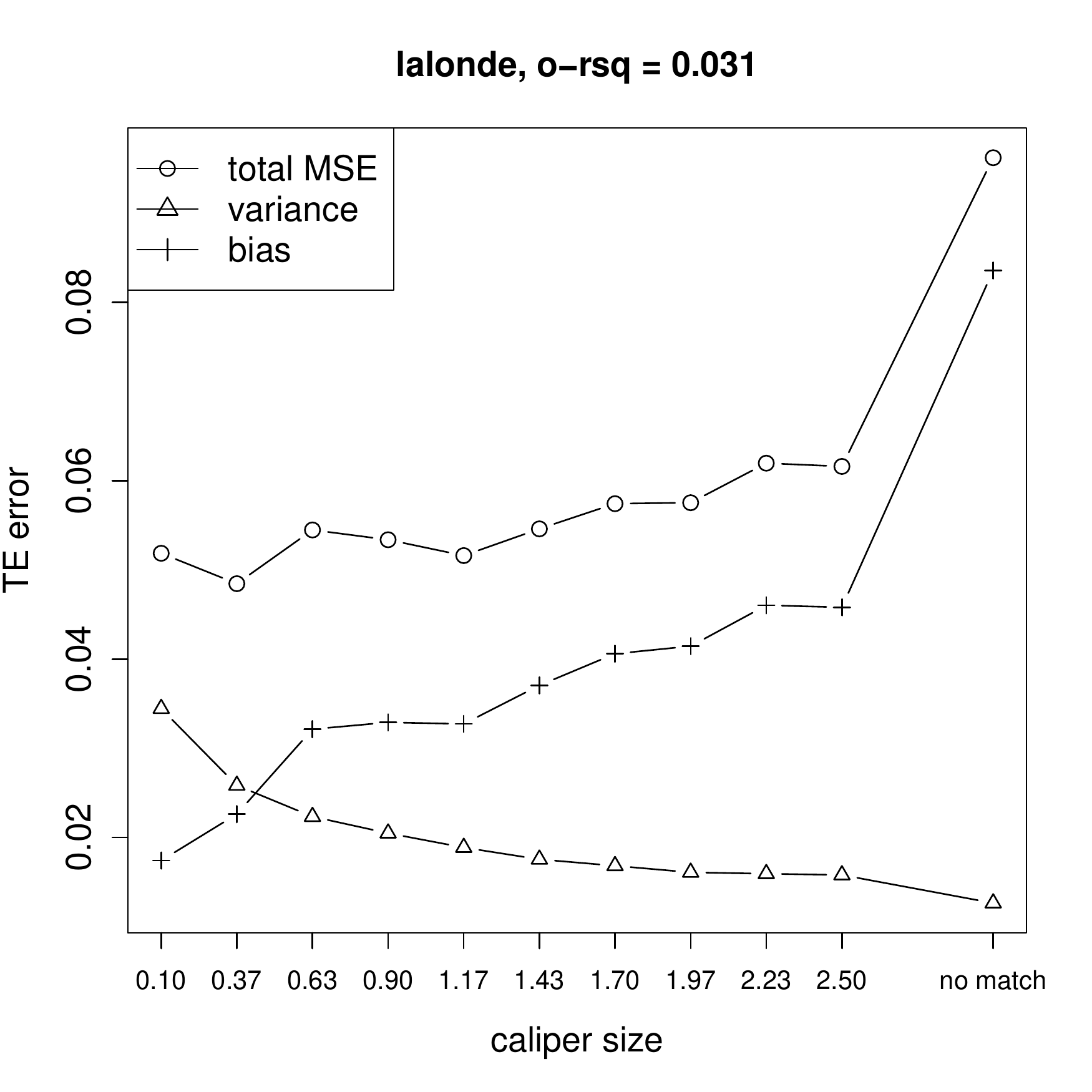}} & 
\subfloat{\includegraphics[height=2.in]{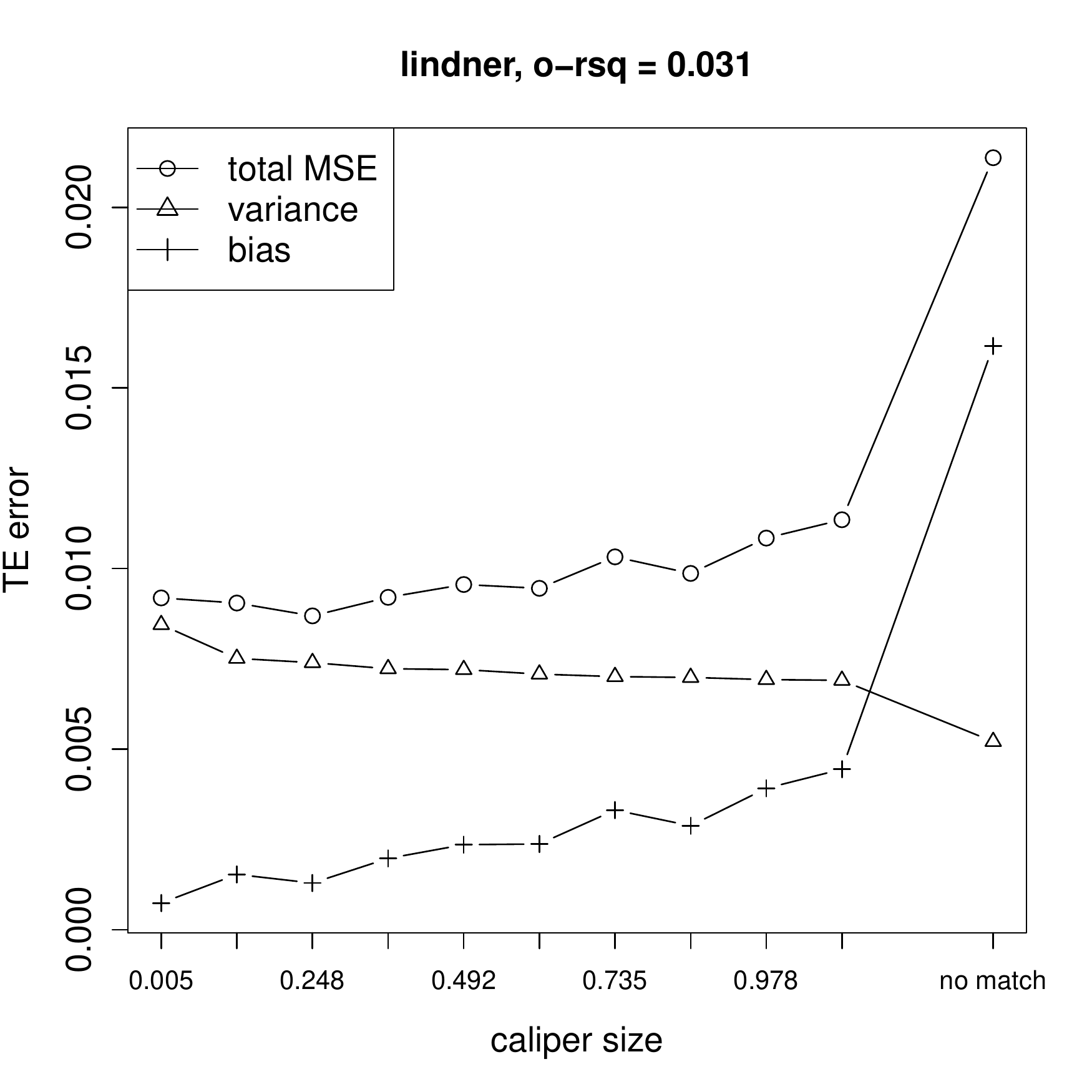}} \\
\subfloat{\includegraphics[height=2.in]{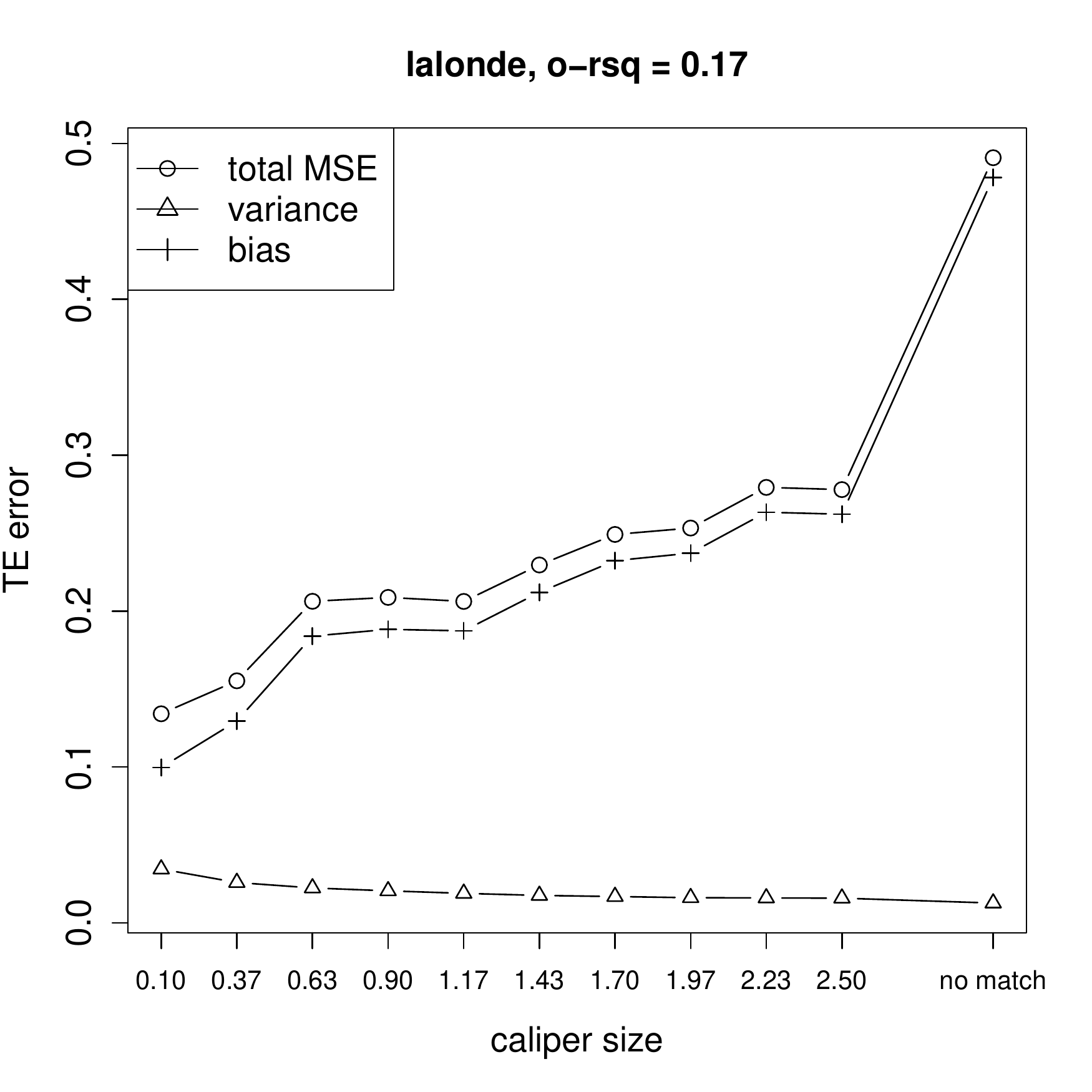}} & 
\subfloat{\includegraphics[height=2.in]{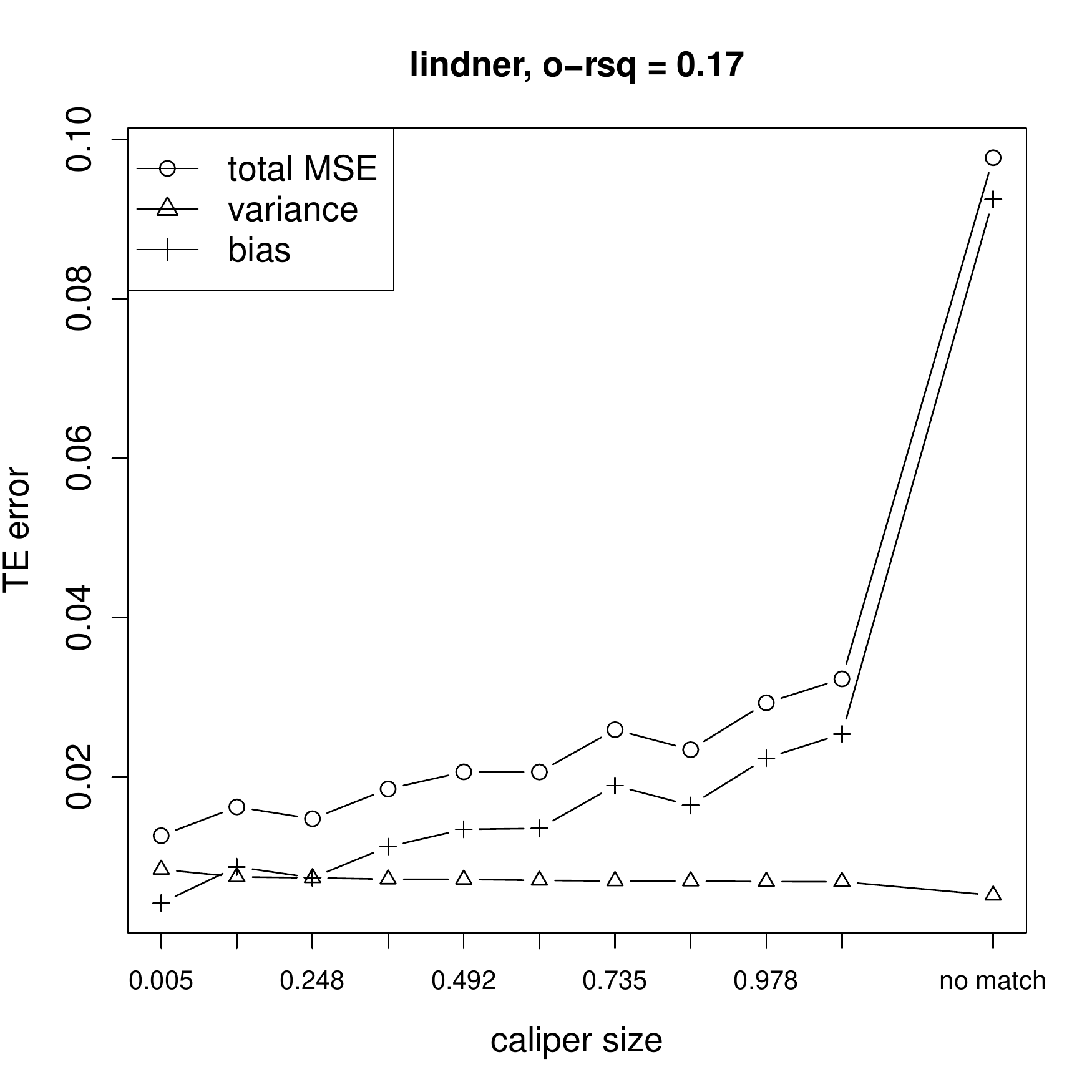}} \\
\subfloat{\includegraphics[height=2.in]{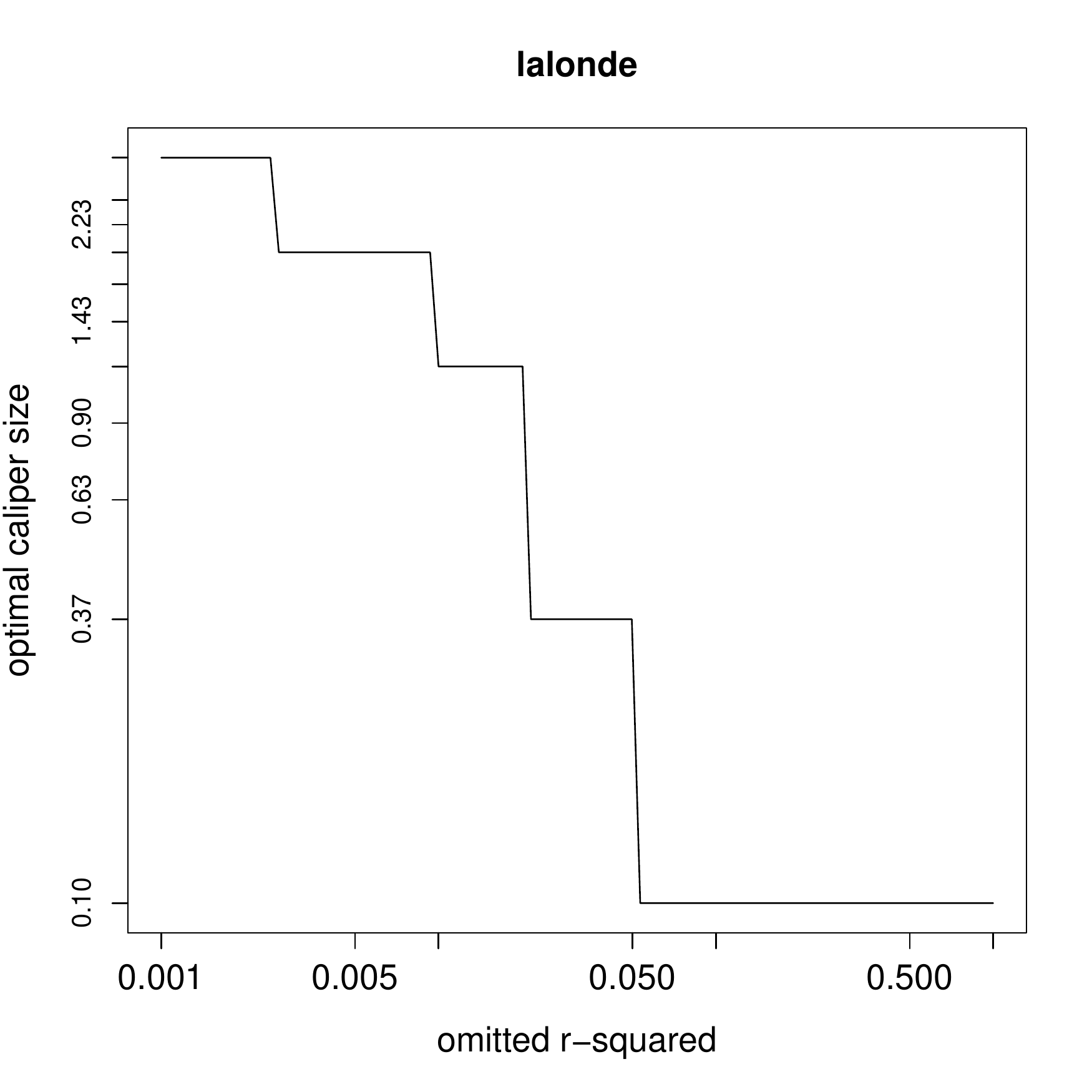}} & 
\subfloat{\includegraphics[height=2.in]{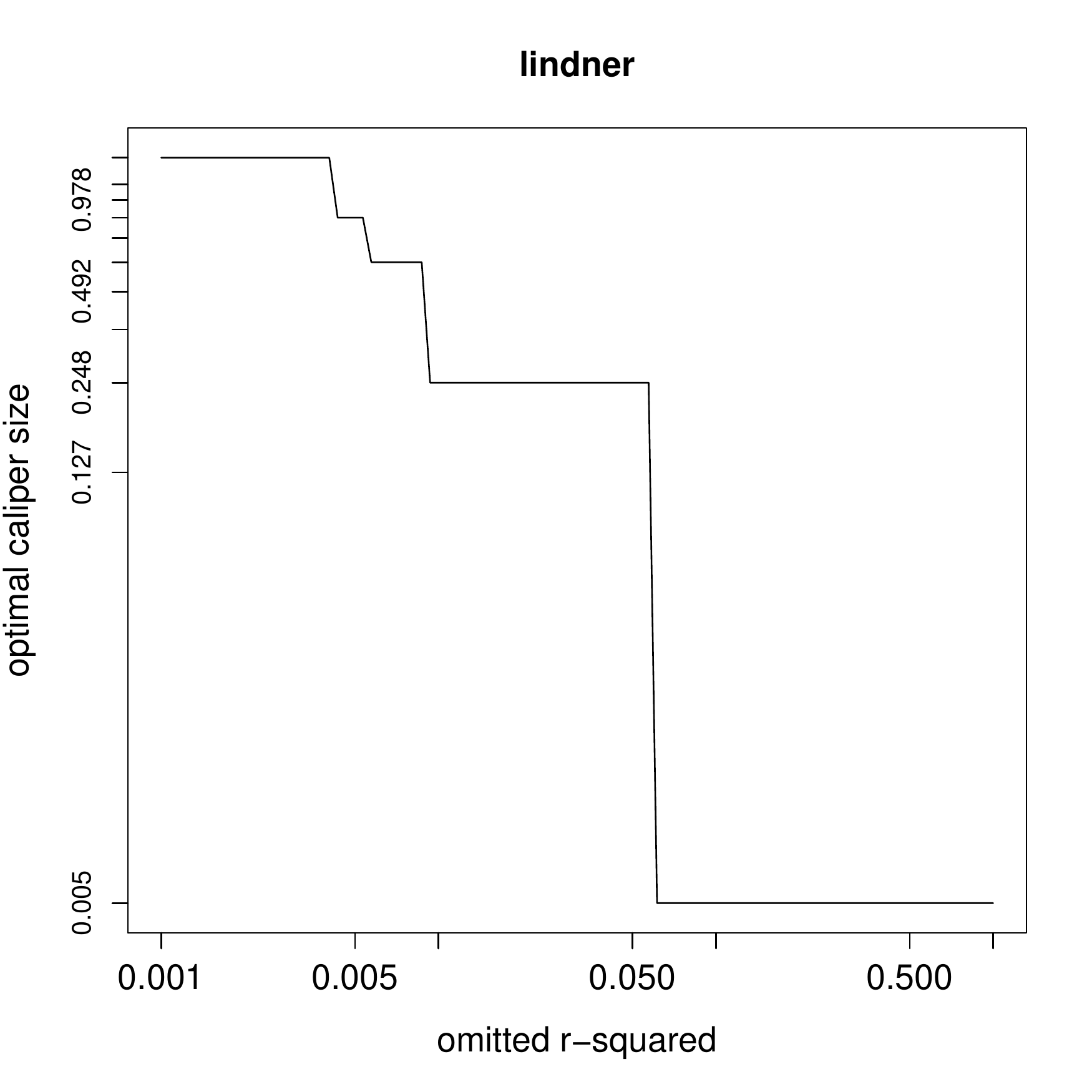}}
\end{tabular}
\caption{Matching calibration by combining bias and variance and minimizing MSE, for lalonde (left) and lindner (right) data sets. Top three rows: TE error and its components as a function of matching caliper size, for three levels of omitted R-squared. Bottom row: Optimal caliper size as a function of omitted R-squared.}
\label{fig-combine-bias-variance}
\end{figure}

How should practitioners go about selecting a value (or a range of values) for omitted R-squared? The answer is, partly domain expertise, and partly experience. In particular, if outcome data is available, practitioners can calculate R-squared for a main-effect-only regression model, and then reason about how much they expect nonlinear effects to improve upon this R-squared. For example, assume that the main-effect R-squared is 14\%, and that we believe nonlinear effects contribute about 30\% towards explaining data variance. Then we can reasonably assume that $R_o^2 = 30\% \times 14\% \sim 4\%$. Some researchers and practitioners may approve of such a limited use of outcome data.

\subsection{Power analysis}\label{subsection-power-analysis}
In covariate-subspace maximization approach to constrained bias estimation (Section~\ref{subsubsection-constrained-bias-estimation}), we erred on using high bias estimation in our bias-variance trade-off analysis by using the maximum operator on the eligible subspace. In this section, we take the opposite view by assuming no bias, and assess the negative impact of matching in terms of variance increase, or equivalently study power decrease.

Figure~\ref{fig-power-analysis} shows study power as a function of caliper size for lalonde and lindner data sets. Effect size is defined as ratio of TE to noise standard deviation, i.e. $\tau / \sigma_0$ and set to $0.3$ for both plot. In each case, we used 1000 Monte Carlo simulations to generate mean and standard deviation values for study power. For comparison, we have also plotted study power when data is subsampled randomly with the same sizes as the matched case from treatment and control groups. The significantly higher study power for matched subsamples is an illustration of the minimum-variance property of matching (Theorem~\ref{theorem-variance-no-replacement}).

\begin{figure}
\centering
\begin{tabular}{cc}
\subfloat{\includegraphics[height=3.2in]{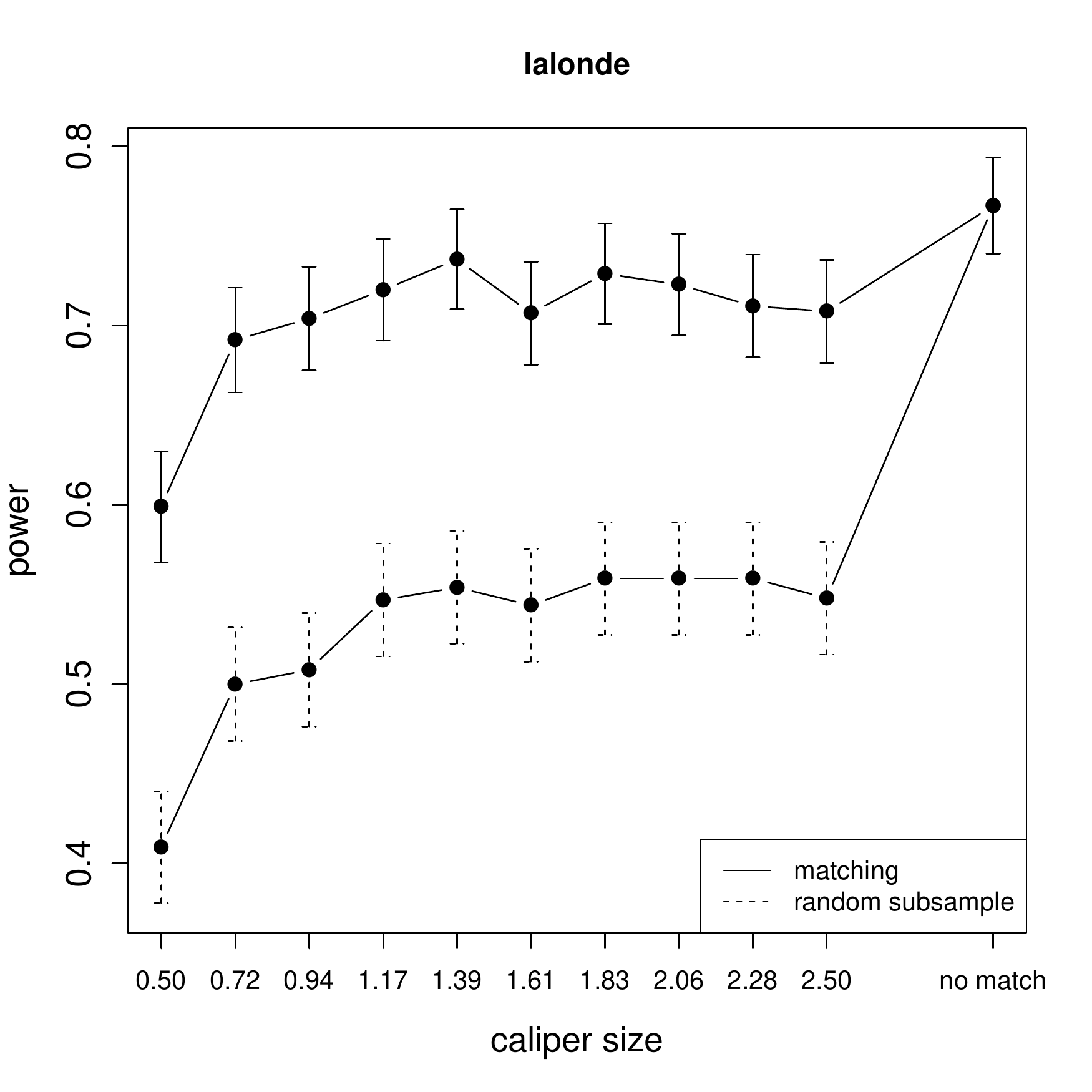}} & 
\subfloat{\includegraphics[height=3.2in]{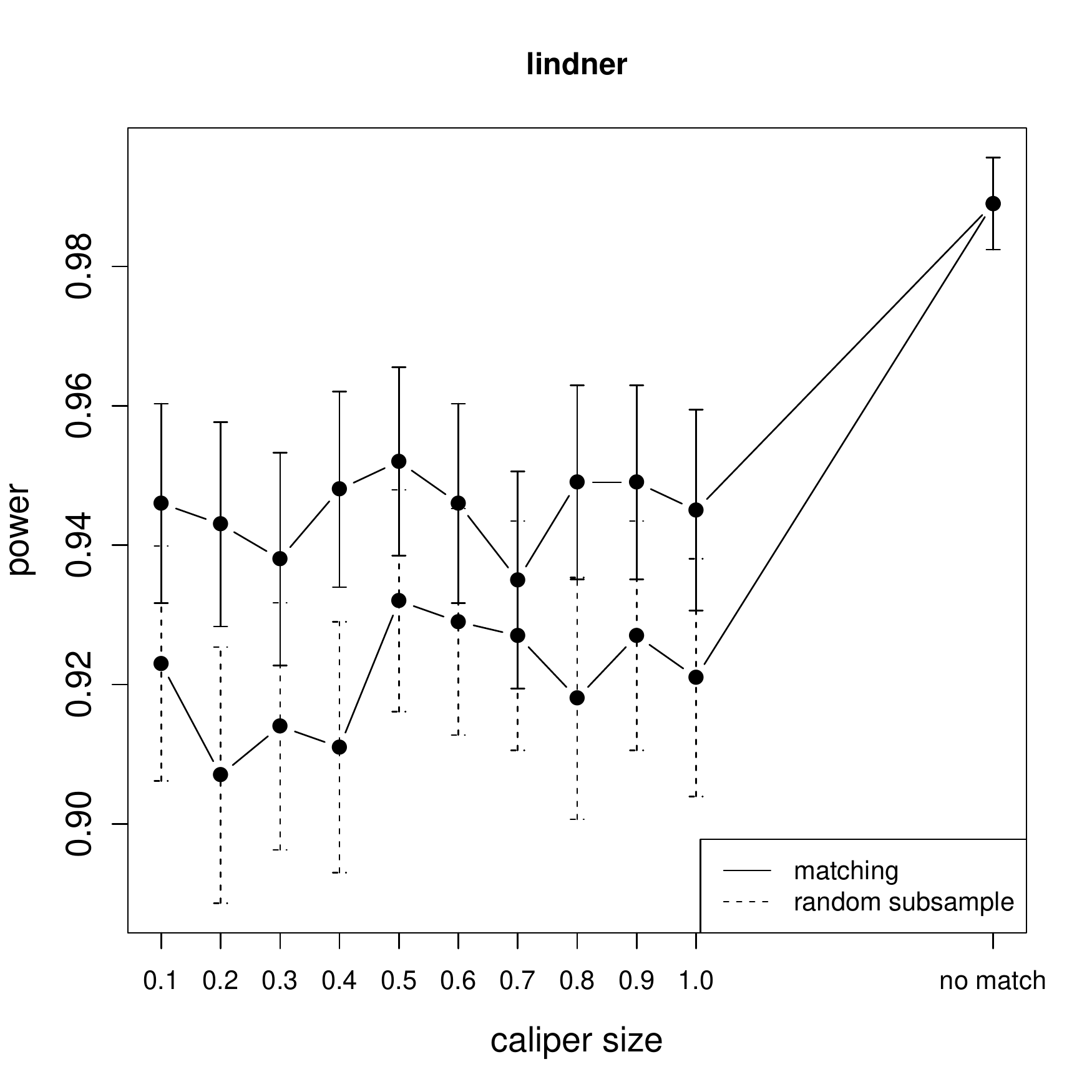}}
\end{tabular}
\caption{Impact of matching caliper size on study power for lalonde (left) and lindner (right) data sets. Effect size is defined as ratio of TE to standard deviation of noise ($\sigma$). Its value was set using regression on data sets, using main effects only. [add numbers for effect size]}
\label{fig-power-analysis}
\end{figure}

[we probably need another sentence or two summarizing what we have done so far. also, elaborate on this MC based approach being the first open-source function specialized in TE estimation and not just general regression.]

\subsection{Software: R package MatchLinReg}\label{subsection-software}
All the diagnostic and calibration tools presented in this section have been packaged into an open-source R library, \textit{MatchLinReg}~\citep{mahani2015matchlinreg}. To allow researchers to develop custom diagnostic tools based on their domain-specific needs, the core set of functions in the package have also been made public, in addition to the higher-level wrappers that implement the practical tools. Below we briefly describe the structure of the software, and refer the interested reader to the package documentation for details. Below we provide a brief description of core and wrapper functions.
\begin{enumerate}
\item Core module: Functions for calculating bias and variance, finding orthogonal and parallel projections of a vector in a subspace, and Monte-Carlo based power calculation for causal inference. Figure~\ref{fig-process-diagram} provides a summary of how MatchLinReg implements our framework for calculation of bias and variance and combining them to produce MSE for TE.
\item Diagnostic/calibration module: Functions for propensity score and Mahalanobis matching (wrapper around \textit{Matching} R package), combining bias and variance over a vector of value for omitted R-squared, and performing diagnostic and calibration analyses described in this section.
\end{enumerate}

\begin{figure}
\centering
\begin{tabular}{cc}
\subfloat{\includegraphics[height=4.5in]{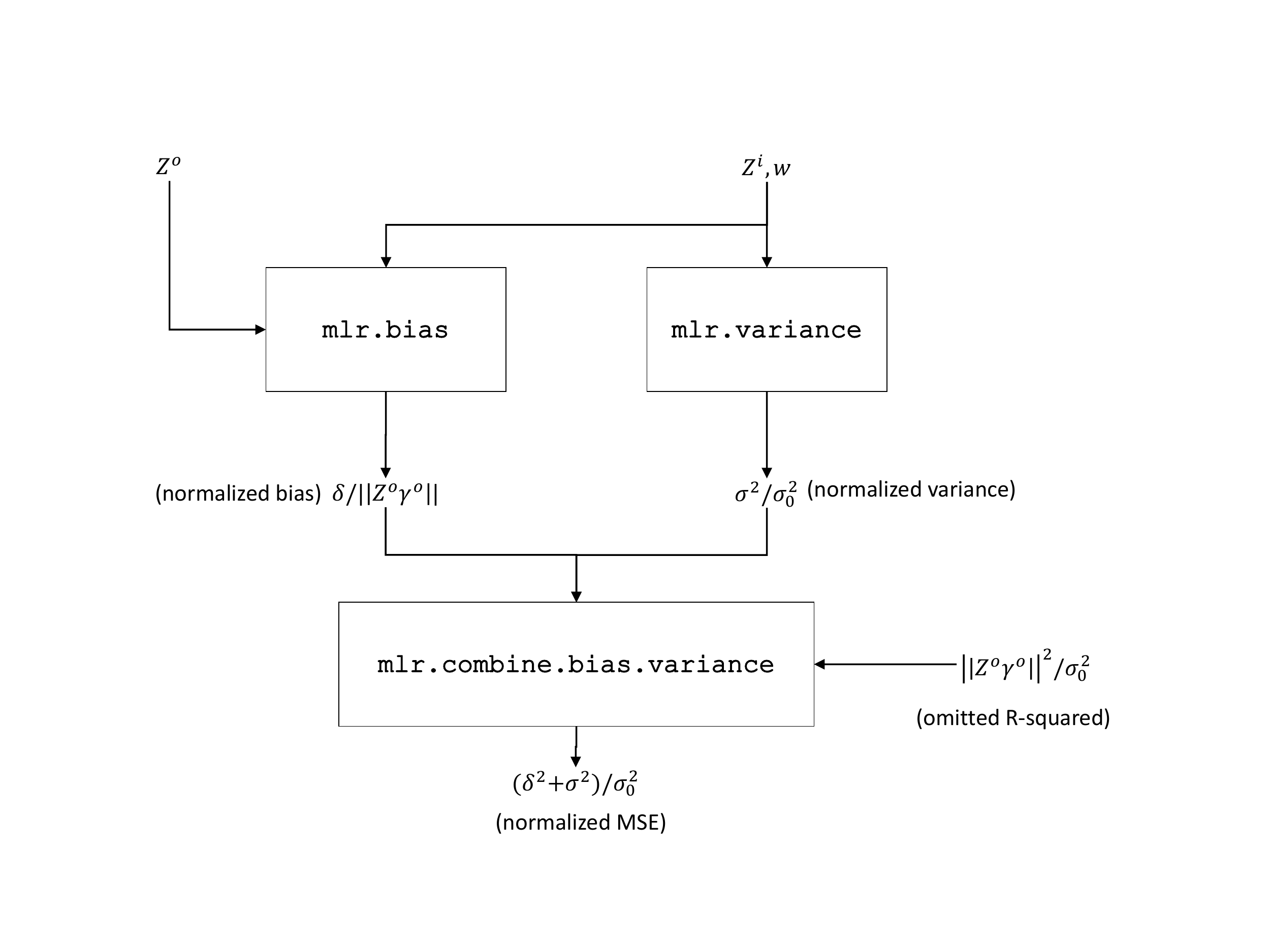}}
\end{tabular}
\caption{Process diagram for calculating total MSE for TE, using MatchLinReg R package. Boxes represent library functions. [describe symbols]}
\label{fig-process-diagram}
\end{figure}

\section{Discussion}\label{section-discussion}

\subsection{Summary}\label{subsection-summary}

In this paper, we developed a framework for quantifying the combined effect of matching - as a non-parametric pre-processing technique - and linear regression for estimating TE for causal inference. In addition to providing a theoretical basis for the impact of matching on TE bias and variance in a misspecified linear regression, the applicability of our framework is strengthened by four attributes: 1) a finite-sample focus (as opposed to large-sample or asymptotic analysis) makes our results directly applicable to real-world problems with small control and/or treatment groups, 2) quantifying not just bias but also variance of TE estimation allows us to minimize total MSE, which is the ultimate measure of estimation accuracy in a single experiment, 3) exclusion of the outcome variable from the equations ensures that the use of our diagnostic and calibration tools by the broad community of practitioners does not lead to overfit results that have questionable prognostic value [move overfitting to beginning of sentence], and 4) capturing the entire framework along with calibration and diagnostic tools in an open-source software, R package MatchLinReg, allows researchers and practitioners to utilize the currently provided set of tools for their observational research, extend and experiment with the framework, and implement new diagnostic and calibration functions that better suit their application domains.

\subsection{Future research}\label{subsection-future-research}

Some of the limitations and assumptions of the current work provide for interesting research opportunities. Below we highlight two key areas:
\begin{enumerate}
\item Generalized Linear Models (GLMs): The current framework is focused on, and takes full advantage of, linear models with a continuous (unbounded) outcome variable. This facilitated the derivation of closed-form expressions for bias and variance. Maximum-Likelihood (ML) estimation of GLM coefficients does not generally contain a closed-form solution. The nonlinear link function creates complex interactions between the coefficients of adjustment covariates (not just omitted but also included) and TE bias and variance. [add literature on matching, bias, conditional vs marginal]
\item Large regressions as alternative to matching: An alternative to matching for removing covariate omission bias is to simply include all candidate omitted terms in the regression model. Any nuisance covariates, i.e. those not present in generative model but included in the regression model, do not contribute to bias - as we can simply assume their corresponding coefficients are zero - but only increase variance. It may appear that increased variance due to inclusion of all candidate terms may overwhelm the regression, thus becoming an unrealistic candidate. However, our early simulations and analysis suggest that this may not be the case, and this topic deserves further research.
\end{enumerate}

\bibliographystyle{rss}
\bibliography{mlr}

\appendix

\section{List of mathematical symbols}\label{appendix-math-symbols}

\begin{enumerate}
\item $\X$: Design matrix, matrix of covariates, or matrix of explanatory variables (including intercept and TE). No distinction is made between included and omitted covariates.
\item $\one$: Vector of length $N$, with all elements equal to 1. This corresponds to the `intercept' covariate.
\item $\tr$: Binary vector of length $N$, indicating whether an observation belongs to treatment group (1) or control group (0); also known as `treatment indicator'.
\item $\Z$: Matrix of adjustment covariates, which is the full design matrix with $\one$ and $\tr$ columns removed: $\X = \left[ \begin{array}{ccc} \tr & \one & \Z \end{array} \right]$.
\item $\Z^i / \Z^o$: Subsets of $\Z$, consisting of adjustment covariates included in / omitted from the regression model: $\Z = \left[ \begin{array}{cc} \Z^i & \Z^o \end{array} \right]$.
\item $\X^i$: Matrix of included covariates (which always includes $\one$ and $\tr$): $\X = \left[ \begin{array}{cc} \X^i & \Z^o \end{array} \right]$.
\item $\x_{n,}$: A column vector, extracted from the $n$'th row of $\X$: $\X = \left[ \begin{array}{ccc} \x_{1,} & \hdots & \x_{N,} \end{array} \right]^t$; similarly defined for $\Z$.
\item $\x_{,k}$: A column vector, extracted from the $k$'th column of $\X$: $\X = \left[ \begin{array}{ccc} \x_{,1} & \hdots & \x_{,K} \end{array} \right]$; similarly defined for $\Z$.
\item $x_{nk}$: A scalar, $k$'th element of $\x_n$; similarly defined for $\Z$.
\item $\y$: Response (column) vector, assumed to be continuous (linear regression).
\item $N/N_t/N_c$: Number of observations in entire data set / treatment group / control group. $N$ must be equal to length ofr $\y$ and number of rows in $\X$. We also have $N = N_t + N_c$.
\item $\beta_0$: Intercept term, or coefficient of covariate $\one$.
\item $\trn$: $n$'th element of $\tr$.
\item $\alphaa$: TE, or coefficient of covariate $\tr$.
\item $\bbeta$: Vector of all coefficients, including intercept and TE.
\item $\ggamma$: Vector of coefficients for adjustment covariates ($\Z$).
\item $\halphaa / \hbetaz / \hbeta / \hgamma$: Counterparts to $\alphaa / \beta_0 / \bbeta / \ggamma$, but represent estimated values from Ordinary Least Squares (OLS) estimator.
\item $T/C$: Set of observation indexes belonging to treatment / control groups: $T \cap C = \emptyset$ and $T \cup C = \{ 1, 2, \hdots, N\}$.
\item $\rrho / \rrho^i / \rrho^o$: Vector of correlation coefficients between all/included/omitted adjustment covariates, and treatment indicator vector: $ \rrho \equiv [ \, \corr(\tr, \z_{,k}) \, ]_{k=1,\hdots,K}$, and similarly for $\rrho^i$ and $\rrho^o$.
\item $\PPhi^i$: Correlation matrix for included adjustment covariates: $\PPhi^i \equiv [\, \corr(\z^i_{,k_1}, \z^i_{,k_2}) \,]_{k_1,k_2=1,\hdots,K}$.
\item $\PPhi^{io}$: Correlation matrix between included and omitted adjustment covariates: \\ $\PPhi^{io} \equiv [\, \corr(\z^i_{,k_1}, \z^o_{,k_2}) \, ]_{k_1=1,\hdots,K', k_2=1,\hdots,K}$.
\item $\wcm$: `weighted correlation matrix' for the included covariates.
\item $\mmu / \mmu^i / \mmu^o$: Vector of means for all/included/omitted adjustment covariates.
\item $\SSigma / \SSigma^i / \SSigma^o$: Matrix of diagonal elements of covariance matrix for all/included/omitted adjustment covariates: $\SSigma \equiv [\, \cov(\z_{,k_1}, \z_{,k_2}) \, \delta_{k_1,k_2} \,]_{k_1,k_2=1,\hdots,K'-K}$, and similarly for $\SSigma^i$ and $\SSigma^o$.
\item $\uu / \uu^i / \uu^o$: Data balance vector for all/included/omitted adjustment covariates, defined as
\begin{equation}
\uu \equiv \frac{1}{N_c} \, \Z^{t} \one - (\frac{1}{N_t} + \frac{1}{N_c}) \, \Z^{t} \, \tr,
\end{equation}
and similarly for $\uu^i$ and $\uu^o$.
\item $\eepsilon$: Vector of length $N$, representing the residual noise in generative linear model.
\item $\PPsi$: Noise covariance matrix: $\PPsi \equiv \cov(\eepsilon)$.
\item $\bias$: TE estimation bias, i.e. $\bias \equiv \E[\halpha] - \alphaa$.
\item $\variance$: TE estimation variance, i.e. $\variance \equiv \E[ (\halpha - \E[\halpha])^2 ]$.
\item $\dd / \dd^i / \dd^o$: Vector of standardized mean differences for all/included/omitted adjustment covariates.
\end{enumerate}
Notes:
\begin{enumerate}
\item For notational brevity, we use symbols $\alphaa / \beta_0 / \bbeta / \ggamma$ to represent both variables that are optimized in error function, and the true, generative model used to produce the data. The distinction should be clear from context, and is clarified in the text if not.
\item Covariates are arranged in the full design matrix in a particular order to simplify representation: $\X = \left[ \begin{array}{cccc} \tr & \one & \Z^i & \Z^o \end{array} \right]$. This does not reduce generality of results.
\item Similarly, without loss of generality, we assume that all treatment observations occupy the first $N_t$ rows of $\X$ (and other corresponding structures), and control observations occupy the last $N_c$ rows of $\X$.
\end{enumerate}

\section{List of acronyms}\label{appendix-acronyms}
\begin{enumerate}
\item ANCOVA: Analysis of Covariance
\item ATC: Average treatment effect for controls
\item ATE: Sample treatment effect
\item ATT: Average treatment effect for the treated
\item MSE: Mean Squared Error
\item OLS: Ordinary Least Squares
\item PSM: Propensity Score Matching
\item SMD: Standardized Mean Difference
\item SLM: Standard Linear Model
\item TE: Treatment Effect
\end{enumerate}

\section{Calculating the first row/column of $\xtxi$}\label{appendix-xtxi}
Assume that
\begin{equation}\label{eq-xtxi-raw}
\xtxi = \left[ \begin{array}{ccc} a & b & \vv^t \\ b & \hdots & \hdots \\ \vv & \hdots & \hdots \end{array} \right].
\end{equation}
Using Eqs.~\ref{eq-xtx}, \ref{eq-def-p} and \ref{eq-def-q}, and applying the definition of matrix inversion, we obtain:
\begin{equation}
\left[ \begin{array}{ccc} N_t & N_t & \pp^t \\ N_t & N_t + N_c & \qq^t \\ \pp & \qq & \RR \end{array} \right]
 \left[ \begin{array}{ccc} a & b & \vv^t \\ b & \hdots & \hdots \\ \vv & \hdots & \hdots \end{array} \right] = 
\left[ \begin{array}{ccc} 1 & \hdots & \hdots \\ 0 & \hdots & \hdots \\ \zero & \hdots & \hdots \end{array} \right],
\end{equation}
where we have defined:
\begin{eqnarray}
\RR &\equiv& \xtxc.
\end{eqnarray}
Next, we expand the expressions producing the first column of the right-hand side (identity matrix) to get:
\begin{equation}\label{eq-xtxi-system}
\left\{
\begin{array}{l} N_t \, a + N_t \, b  +\pp^t \, \vv = 1, \\ N_t \, a + (N_t + N_c) \, b + \qq^t \, \vv = 0, \\ \pp \, a + \qq \, b + \RR \, \vv = \zero. \end{array}
\right.
\end{equation}
Solving for $a$ and $b$ (in terms of $\vv$) between the top 2 sub-equations in \ref{eq-xtxi-system} leads to:
\begin{eqnarray}
a &=& \frac{1}{N_t} + \frac{1}{N_c} + \frac{1}{N_t N_c} \left( N_t \, \qq - (N_t + N_c) \, \pp \right)^t \vv, \label{eq-a-incomplete} \\
b &=& -\frac{1}{N_c} + \frac{1}{N_c} \left( \pp - \qq \right)^t \vv.
\end{eqnarray}
Substituting the above back into the last sub-equation of \ref{eq-xtxi-system}, we get:
\begin{equation}
\vv = \frac{1}{N_c \, N_t} \A^{-1} \left( N_t \, \qq - (N_t + N_c) \, \pp \right),
\end{equation}
where
\begin{equation}\label{eq-def-A-original}
\A \equiv \xtxc + \frac{1}{N_c} \pp \, \qq^t - (\frac{1}{N_t} + \frac{1}{N_c}) \, \pp \, \pp^t  + \frac{1}{N_c} \, \qq \, \pp^t  - \frac{1}{N_c} \, \qq \, \qq^t.
\end{equation}
Substituting the above solution for $\vv$ back into Equation~\ref{eq-a-incomplete}, we obtain the following expression:
\begin{align}
a &= \frac{1}{N_t} + \frac{1}{N_c} + \uu^{i,t} \A^{-1} \uu^i, \label{eq-a-complete} \\
b &= -\frac{1}{N_c} + \frac{1}{N_c} \, (\pp - \qq)^t \A^{-1} \, \uu^i, \\
\vv &= \A^{-1} \, \uu^i,
\end{align}
where $\uu^i$ is defined in Eq.~\ref{eq-def-ui}. Putting it all together, we get the following:
\begin{align}
\xtxi = \left( \begin{array}{ccc} \frac{1}{N_t} + \frac{1}{N_c} + \uu^{i,t} \A^{-1} \uu^i & -\frac{1}{N_c} + \frac{1}{N_c} \, (\pp - \qq)^t \A^{-1} \, \uu^i & \uu^{i,t} \A^{-1} \\ -\frac{1}{N_c} + \frac{1}{N_c} \, (\pp - \qq)^t \A^{-1} \, \uu^i & \hdots & \hdots \\  \A^{-1} \, \uu^i & \hdots & \hdots \end{array} \right),
\end{align}
where we have taken advantage of $\A$ being symmetric. Arriving at Eq.~\ref{eq-def-A} involves some routine algebraic manipulation of Eq.~\ref{eq-def-A-original} while taking advantage of the definitions of within-group covariance matrices in Eqs.~\ref{eq-cov-T} and \ref{eq-cov-C}.

Note that, when included covariates are balanced ($\uu^i = \zero$), then the above expression becomes:
\begin{align}
\xtxi = \left( \begin{array}{ccc} \frac{1}{N_t} + \frac{1}{N_c} & -\frac{1}{N_c} & \zero \\ -\frac{1}{N_c} & \hdots & \hdots \\  \zero & \hdots & \hdots \end{array} \right). \label{eq-xtxi-balanced}
\end{align}

\section{Normalized bias and variance equations}\label{appendix-normalized-eqs}
We seek to derive the normalized expressions for bias and variance (Eqs.~\ref{eq-bias-normal} and \ref{eq-variance-normal}) from the standard forms (Eqs.~\ref{eq-bias-standard} and \ref{eq-variance-standard}). To transform the variance equation, we first prove Eq.~\ref{eq-rrho-u}, showing that the the vector of correlations between treatment indicator ($\tr$) and included adjustment covariates ($\Z^i$) is proportional to the vector of mean differences for $\Z^i$.

Focusing on a single covariate vector ($\z$) of length $N$, its correlation ($\rho$) with treatment indicator ($\tr$) is given by:
\begin{align}
\rho &= \frac{1}{(N-1) \, \sigma_z \, \sigma_w} \sum_{n=1}^N (z_n - <\z>)(w_n - <\tr>), \\
&= \frac{1}{(N-1) \, \sigma_z \, \sigma_w} \{ N <\tr^t \z> - N <\z> <\tr> \}. \label{eq-def-rho}
\end{align}
Taking advantage of the definition of $\tr$, it is easy to verify that
\begin{align}
<\tr^t \z> &= \frac{N_t}{N} <\z>_T, \label{eq-mean-wz} \\
<\tr> &= \frac{N_t}{N}, \label{eq-mean-w} \\
\sigma_w^2 &= \frac{N_t N_c}{N (N - 1)}. \label{eq-sig-w}
\end{align}
Combining Eqs.~\ref{eq-def-rho}, \ref{eq-mean-wz}, \ref{eq-mean-w} and \ref{eq-sig-w}, we obtain:
\begin{align}
\rho &= \frac{N^{1/2} (N - 1)^{1/2}}{(N - 1) \sigma_z N_t^{1/2} N_c^{1/2}} \{ N \frac{N_t}{N} <\z>_T - N \frac{N_t}{N} <\z> \} \\
&= (\frac{N N_t}{(N - 1) N_c \sigma_z^2})^{1/2} \{ <\z>_T - <\z> \}. \label{eq-rho-semi}
\end{align}
On the other hand, from the following self-evident equality:
\begin{align}
N_C <\z>_C + N_t <\z>_T = N <\z>,
\end{align}
we conclude that
\begin{align}
<\z>_T - <\z> &= <\z>_T - \{ \frac{N_c}{N} <\z>_C + \frac{N_t}{N} <\z>_T \}, \\
&= \frac{N_c}{N} \{ <\z>_T - <\z>_C \}. \label{eq-zt-zc}
\end{align}
Combining Eqs.~\ref{eq-rho-semi} and \ref{eq-zt-zc}, we get:
\begin{align}
\rho &= (\frac{N N_t}{(N - 1) N_c \sigma_z^2})^{1/2} \frac{N_c}{N} \{ <\z>_T - <\z>_C \} \\
&= (\frac{N_t N_c}{N (N - 1)})^{1/2} \{ <\z>_T - <\z>_C \} / \sigma_z \\
\end{align}
Eq.~\ref{eq-rrho-u} is simply the vector form of the above expression.

To obtain the normalized variance equation, we invert Eq.~\ref{eq-rrho-u} to express $\uu^i$ in terms of $\rrho^i$, and substitute into Eq.~\ref{eq-variance-standard}:
\begin{align}
\sigma^2 &= \sigma_0^2 \left( \frac{1}{N_t} + \frac{1}{N_c} + \frac{N (N - 1)}{N_t N_c} {\SSigma^i}^{1/2} \rrho^{i,t} \A^{-1} {\SSigma^i}^{1/2} \rrho^i \right), \\
&= \sigma_0^2 \left( \frac{1}{N_t} + \frac{1}{N_c} + \frac{N_t + N_c}{N_t N_c} \rrho^{i,t} \{ {\SSigma^i}^{-1/2} \, \A \, {\SSigma^i}^{-1/2} \}^{-1} \rrho^i \right),
\end{align}
where, in the last step, we have used the matrix inversion lemma, $(\A \, \B)^{-1} = \B^{-1} \, \A^{-1}$. Using the definition of $\wcm$ in Eq.~\ref{eq-def-wcm}, the above expression can be easily turned into Eq.~\ref{eq-variance-normal}.

To derive the normalized bias expression of Eq.~\ref{eq-bias-normal}, we must express its constituent entities, i.e. $\pp - \qq$, $\Z^{o,t} \tr$, $\Z^{o,t} \one$ and $\Z^{o,t} \Z^i$, in normalized forms. (We already have an expression for $\uu^i$ in Eq.~\ref{eq-rrho-u}.) The reader can easily verify that:
\begin{align}
\pp - \qq &= - N_c \, \mmu^i + \left( \frac{N_t N_c (N - 1)}{N} \right)^{1/2} {\SSigma^i}^{1/2} \rrho^i \\
\Z^{o,t} \one &= N \, \mmu^o \\
\Z^{o,t} \tr &= N_t \, \left\{ \mmu^o + \left( \frac{(N - 1) N_c}{N N_t} \right)^{1/2} \, {\SSigma^o}^{1/2} \rrho^o \right\} \\
\Z^{o,t} \Z^i &= (N - 1) \, {\SSigma^o}^{1/2} \, \PPhi^{io} \, {\SSigma^i}^{1/2} + N \, \mmu^o \, \mmu^{i,t}
\end{align}
From the above, we obtain the three additive components of bias (before multiplying $\ggamma^o$):
\begin{align}
\left( \frac{1}{N_t} + \frac{1}{N_c} + \uu^{i,t} \, \A^{-1} \, \uu^i \right) \, \Z^{o,t} \tr &= N_t \, \left( \frac{1}{N_t} + \frac{1}{N_c} \right) \left( 1 + \rrho^{i,t} \wcm^{-1} \rrho^i \right) \times \nonumber \\
& \left\{ \mmu^o + \left( \frac{(N - 1) N_c}{N N_t} \right)^{1/2} \, {\SSigma^o}^{1/2} \, \rrho^o \right\} \\
\left( -\frac{1}{N_c} + \frac{1}{N_c} (\pp - \qq) \, \A^{-1} \, \uu^i \right) \, \Z^{o,t} \one &= N \Bigg\{ - \frac{1}{N_c} + \left( \frac{N}{N_t N_c (N - 1)} \right)^{1/2} \, \mmu^{i,t} \, {\SSigma^i}^{-1/2} \, \wcm^{-1} \, \rrho^i \nonumber \\
&- \frac{1}{N_c} \, \rrho^{i,t} \, \wcm^{-1} \, \rrho^i \Bigg\} \, \mmu^o \\
\Z^{o,t} \, \Z^i \, \A^{-1} \, \uu^i &= - \left( \frac{N (N - 1)}{N_t N_c} \right)^{1/2} {\SSigma^o}^{1/2} \, \PPhi^{io} \, \wcm^{-1} \, \rrho^i \nonumber \\
&- \left( \frac{N^3}{N_t N_c (N - 1)} \right)^{1/2} \mmu^o \, \mmu^{i,t} \, {\SSigma^i}^{-1/2} \, \wcm^{-1} \, \rrho^i
\end{align}
Close examination of the above three equations reveals that all the terms involving $\mmu^o$ cancel each other out, leading us to Eq.~\ref{eq-bias-normal}.

\section{Invariance of bias estimation under matching using absolute maximization}\label{appendix-absolute-matching}
Before matching, the direction of $\Z^o \ggamma^o$ with maximum absolute bias is $\g$, and TE bias is given by $\g^t (\sqrt{N} \, \g / ||\g||)$, where $\g$ is the top row of matrix $\UU$ defined as
\begin{align}
\UU \equiv \xtxi \, \X^{i,t}. \label{eq-def-U}
\end{align}
In other words, $\g^t \g$ is the top-left element of the matrix $\UU \UU^t$. After matching, we do not re-normalize $\Z^o \ggamma^o$ [explain why, here or better in main text]. Without loss of generality, assume that matching subselects a contiguous subset ($\X_1$) of $\X^i$:
\begin{align}
\X^i = \left[ \begin{array}{c} \X_1 \\ \X_2 \end{array} \right]. \label{eq-xt-split}
\end{align}
Inserting Eq.~\ref{eq-xt-split} into \ref{eq-def-U}, we obtain:
\begin{align}
\UU &= \xtxi \left[ \begin{array}{cc} \X_1^t & \X_2^t \end{array} \right].
\end{align}
Therefore,
\begin{align}
\UU^t &= \left[ \begin{array}{c} \X_1 \\ \X_2 \end{array} \right] \xtxi, \\
&= \left[ \begin{array}{c} \X_1 \, \xtxi \\ \X_2 \, \xtxi \end{array} \right] \\
&= \left[ \begin{array}{c} \UU_s^t \\ \X_2 \, \xtxi \end{array} \right], \label{eq-Ut}
\end{align}
where $\UU_s^t$ is the matched subset of $\UU^t$:
\begin{align}
\UU_s^t \equiv \X_1 \xtxi.
\end{align}
After matching, bias is given by the top-left element of $\UU_1 \, \UU_s^t$, where $\UU_1$ is given by:
\begin{align}
\UU_1 \equiv (\X_1^t \, \X_1)^{-1} \, \X_1^t \label{eq-def-U1}.
\end{align}
Combining the definitions of $\UU_s^t$ and $\UU_1$ we get:
\begin{align}
\UU_1 \UU_s^t &= (\X_1^t \, \X_1)^{-1} \X_1^t \, \X_1 \xtxi, \\
&= \xtxi.
\end{align}
We see that $\UU_1 \UU_s^t$ is independent of matched subset, thus remaining constant as long as $\X_1^t \, \X_1$ is full rank. We also observe that maximum normalized squared bias can be written as:
\begin{align}
{\bias^n_{max}}^2 = N \, \frac{\sigma^2}{\sigma_0^2}.
\end{align}

\section{Matching with replacement and TE variance}\label{appendix-var-repeat}
Her we prove that, in a balanced data set, TE variance is larger than the lower-bound established in Theorem~\ref{theorem-variance-no-replacement} if some observations are repeated. In sampling (matching) with replacement, the covariance matrix $\cov[\y]$ is not diagonal, and the expression for covariance matrix of coefficients becomes more involved:
\begin{equation}
\cov[\hbeta] = \cov \left[ (\X^{i,t} \X^i)^{-1} \X^{i,t} \y \right] = \DD \, \PPsi \, \DD^t, \label{eq-cov-long}
\end{equation}
with
\begin{align}
\DD &\equiv (\X^{i,t} \X^i)^{-1} \X^{i,t} \label{eq-def-D}, \\
\PPsi &\equiv \cov[\y].
\end{align}
To obtain TE variance ($\variance$), we need the top-left element of $\cov[\hbeta]$:
\begin{align}
\variance = \cov[\hbeta]_{1,1} = (\DD \, \Sigma \, \DD^t)_{1,1} &= \sum_{n=1}^{N} \DD_{1,n} \, (\Sigma \, \DD^t)_{n,1} \nonumber \\
&= \sum_{n=1}^N \sum_{n'=1}^N \DD_{1,n} \Sigma_{n,n'} \DD_{1,n'} \label{eq-cov-double-sum}
\end{align}
The above equation indicates that we only need to focus on the first row of $\DD$ to calculate $\variance$. From the definition of $\DD$ in Equation~\ref{eq-def-D} and $\X^i$ in Equation~\ref{eq-def-Xi}, and using the special form of $\xtxi$ in Equation~\ref{eq-xtxi-balanced}, we conclude that
\begin{align}
\DD = \left( \begin{array}{ccc} 1/N_t + 1/N_c & -1/N_c & \zero \\ -1/N_c & \hdots & \hdots \\ \zero & \hdots & \hdots \end{array} \right) \left( \begin{array}{c} \tr^t \\ \one^t \\ \Z_i^t \end{array} \right) = \left( \begin{array}{c} (1/N_t + 1/N_c) \tr^t - \one^t / N_c \\ \hdots \\ \hdots \end{array} \right)
\end{align}
Examining the above form reveals the following structure for the first row of $\DD$:
\begin{align}
\DD_{1,n} = \left\{ \begin{array}{l} 1/N_t, \mathrm{\,\,if\,\,} n \in T, \\ -1/N_c, \mathrm{\,\,if\,\,} n \in C. \end{array} \right.
\end{align}
Therefore,
\begin{align}
\DD_{1,n} \DD_{1,n'} = \left\{ \begin{array}{l} 1/N_t^2, \mathrm{\,\,if\,\,} n,n' \in T, \\ 1/N_c^2, \mathrm{\,\,if\,\,} n,n' \in C, \\ -1/(N_t \, N_c), \mathrm{\,\,if\,\,} n \in T, n' \in C. \end{array} \right.
\end{align}
Combining the above with Equation~\ref{eq-cov-double-sum}, we obtain:
\begin{align}
\variance = \sum_{n} \DD_{1,n}^2 \Psi_{n,n} + \mathop{\sum\sum}_{n, n' \neq n} \DD_{1,n} \DD_{1,n'} \Psi_{n,n'}
\end{align}
To simplify the above sum, we recall the structure of $\Psi$: its diagonal elements are all equal to $\sigma_0^2$, and its off-diagonal elements are 0, if observations $n$ and $n'$ are not samples from the same original data point, and $\sigma_0^2$ otherwise. We therefore conclude:
\begin{align}
\variance = \sigma_0^2 \left( \sum_{n \in T} \DD_{1,n}^2 + \sum_{n \in C} \DD_{1,n}^2 \right) &+ \mathop{\sum\sum}_{n \in T, n' \in T, n' \neq n} \DD_{1,n} \DD_{1,n'} \Psi_{n,n'} \nonumber \\ 
&+ \mathop{\sum\sum}_{n \in C, n' \in C, n' \neq n} \DD_{1,n} \DD_{1,n'} \Psi_{n,n'} \nonumber
\end{align}
Defining the indicator function $r[.]$ such that:
\begin{align}
r[n] = m \Longleftrightarrow \text{data point \textit{n} is a sample from the original data point \textit{m}},
\end{align}
allows us to characterize the entries of $\Sigma$ as follows:
\begin{align}
\Psi_{n,n'} = \left\{ \begin{array}{l} \sigma^2, \text{if} \,\, r[n] = r[n'],\\ 0, \text{otherwise}. \end{array} \right.
\end{align}
Putting it all together, we get the following expression for TE variance under sampling with replacement:
\begin{align}
\variance = \sigma_0^2 (\frac{1}{N_t} + \frac{1}{N_c}) &+ \frac{\sigma_0^2}{N_t^2}  \mathop{\sum\sum}_{n \in T, n' \in T, n' \neq n} I[r[n] == r[n']] \nonumber \\
&+ \frac{\sigma_0^2}{N_c^2} \mathop{\sum\sum}_{n \in C, n' \in C, n' \neq n} I[r[n] == r[n']] \\
&\geq \sigma_0^2 (\frac{1}{N_t} + \frac{1}{N_c}).
\end{align}
This proves that the TE variance lower bound applies to matching with replacement as well. ($I[.]$ is the indicator function which returns 1 if argument is true, and 0 otherwise.)

\end{document}